\documentclass{article}
\usepackage{amsmath,amsthm,amssymb}
\usepackage[utf8]{inputenc}
\usepackage{fullpage}
\usepackage{url}
\usepackage{graphicx}
\usepackage{xcolor}
\usepackage{tikz}
\usepackage{optidef}
\usepackage{float}
\usepackage[ruled]{algorithm2e} 
\usetikzlibrary{decorations.pathreplacing}

\newcommand{\V}{\mathcal{V}}
\newcommand{\C}{\mathcal{C}}

\newcommand{\mathsc}[1]{{\normalfont\textsc{#1}}}

\newtheorem{lemma}{Lemma}
\newtheorem{theorem}{Theorem}
\newtheorem{corollary}{Corollary}
\newtheorem{proposition}{Proposition}

\theoremstyle{definition}
\newtheorem{example}{Example}
\newtheorem{definition}{Definition}

\newcommand{\g}{\mathsc{Greedy}}
\renewcommand{\b}{\mathsc{Banzhaf}}
\newcommand{\opt}{\mathsc{Opt}}
\newcommand{\rand}{\mathsc{Rand}}
\newcommand{\alg}{\mathsc{Alg}}

\DeclareMathOperator{\argmax}{arg\,max}
\DeclareMathOperator{\argmin}{arg\,min}
\DeclareMathOperator{\E}{\mathbf{E}}

\renewcommand{\d}{\mathrm{d}}

\newcommand{\F}{\mathsc{Fail}}
\newcommand{\p}{\mathsf{P}}
\newcommand{\np}{\mathsf{NP}}
\newcommand{\zpp}{\mathsf{ZPP}}
\newcommand{\rp}{\mathsf{RP}}
\newcommand{\crp}{\mathsf{coRP}}
\newcommand{\nph}{\mathsf{NP}\text{-}\mathsf{hard}}

\newcommand{\rc}{\mathsc{Regular Max K-Cover}}

\definecolor{color1}{rgb}{0.7, 0.2, 0.2}
\definecolor{color2}{rgb}{0.0, 0.4, 0.0}
\definecolor{color3}{rgb}{0.2, 0.4, 0.7}

\title{Optimal Algorithms for Multiwinner Elections and the Chamberlin-Courant Rule}

\newcommand*\samethanks[1][\value{footnote}]{\footnotemark[#1]}
\author{Kamesh Munagala\thanks{Computer Science Department, Duke University. Email: \texttt{kamesh@cs.duke.edu}, \texttt{zeyu.shen@duke.edu}, \texttt{knwang@cs.duke.edu}.} \and Zeyu Shen\samethanks[1] \and Kangning Wang\samethanks[1]}
\date{}

\begin{document}

\maketitle

\begin{abstract}
We consider the algorithmic question of choosing a subset of candidates of a given size $k$ from a set of $m$ candidates, with knowledge of voters' ordinal rankings over all candidates. We consider the well-known and classic scoring rule for achieving diverse representation: the Chamberlin-Courant (CC) or $1$-Borda rule, where the score of a committee is the average over the voters, of the rank of the best candidate in the committee for that voter; and its generalization to the average of the top $s$ best candidates, called the $s$-Borda rule. 

Our first result is an improved analysis of the natural and well-studied greedy heuristic. We show that greedy achieves a $\left(1 - \frac{2}{k+1}\right)$-approximation to the maximization (or satisfaction) version of CC rule, and a $\left(1 - \frac{2s}{k+1}\right)$-approximation to the $s$-Borda score. This significantly improves the existing submodularity-based analysis of the greedy algorithm that only shows a $(1-1/e)$-approximation. Our result also improves on the best known approximation algorithm for this problem. We achieve this result by showing that the average dissatisfaction score for the greedy algorithm is at most $2\frac{m+1}{k+1}$ for the CC rule, and at most $2s^2 \frac{m+1}{k+1}$ for $s$-Borda. We show these dissatisfaction score bounds are tight up to constants, and even the constant factor of $2$ in the case of the CC rule is almost tight.

For the dissatisfaction (or minimization) version of the problem, it is known that the average dissatisfaction score of the best committee cannot be approximated in polynomial time to within any constant factor when $s$ is a constant (under standard computational complexity assumptions). As our next result, we strengthen this to show that the score of $\frac{m+1}{k+1}$ can be viewed as an \emph{optimal benchmark} for the CC rule, in the sense that it is essentially the best achievable score of any polynomial-time algorithm even when the optimal score is a polynomial factor smaller. We show that another well-studied algorithm for this problem, called the Banzhaf rule, attains this benchmark. 

We finally show that for the $s$-Borda rule, when the optimal value is small, these algorithms can be improved by a factor of $\tilde \Omega(\sqrt{s})$ via LP rounding. Our upper and lower bounds are a significant improvement over previous results, and taken together, not only enable us to perform a finer comparison of greedy algorithms for these problems, but also provide analytic justification for using such algorithms in practice.
\end{abstract}

\section{Introduction}
Multiwinner elections are a classical problem in social choice. In this problem, the goal is to find a set of candidates (or winning committee) of fixed size from voter preferences over the candidates. Indeed, some of the earliest work on the design of voting rules that map individual preferences to a winning committee dates back at least a century~\cite{Thiele}. 

Multiwinner elections clearly arise in choosing a winning parliament in representative democracies. They have also recently found applications in design of systems for making procurement or hiring decisions~\cite{LuB,SkworonFL}, and in participatory budgeting~\cite{goel2019knapsack,pb2}. In these settings, the candidates are products or public projects that provide shared utility to individuals. An entity such as a company or a city government has to decide, based on individual preferences, which of these projects or products to produce subject to a cardinality constraint. 

Much of the work on multiwinner elections has focused on the question of \emph{proportional} or \emph{diverse representation}: How can we choose a winning committee where every voter feels they have some representation? Indeed, classic voting rules such as Proportional Approval Voting (PAV)~\cite{Thiele} or Single Transferable Voting (STV)~\cite{Tideman} explicitly attempt to enforce such representation. 

In this paper, we consider the question of choosing a committee of fixed cardinality $k$ from a set $\C$ of $m$ candidates, when voters express \emph{ordinal rankings} over these candidates. In many applications, including parliamentary democracies or participatory budgeting, it is reasonable to assume voters can compare candidates or projects and hence can rank them ordinally, while they may not be able to articulate cardinal utilities for the same. 

A classic set of objectives for ensuring diverse representation~\cite{BrillFST19} based on ordinal preferences uses the so-called \emph{Borda score}. In the minimization (or dissatisfaction) version, the Borda score of candidate $c$ for voter $v$, denoted $r_v(c)$, is the ordinal rank of $c$ in $v$'s ranking. Here, the top-ranked candidate has Borda score $1$, and the bottom-ranked candidate has score $m$.\footnote{Existing literature also uses a score of $0$ for the best ranked and $m-1$ for the worst ranked candidates. Since our results concern absolute scores, they carry over to this setting by simply subtracting $1$ from the bounds. We use a minimum score of $1$ since it is the more challenging setting for showing hardness results.}  Let $\V$ denote the set of all voters, with $n = |\V|$. Given a committee (that is, a set of candidates) $T$ of size $k$, the $s$-Borda score of this committee (for $s \le k$) is given by
\begin{equation}
\label{eq:sborda}
r_{\V}(T) = \frac{1}{n} \sum_{v \in \V} \left(\min_{Q \subseteq T, |Q| = s} \sum_{c \in Q} r_v(c)\right).
\end{equation}
Throughout the paper, we will denote the minimum possible score as $\opt = \min_{T \subseteq \C, |T| = k} r_{\V}(T)$.

To interpret the above score, for each voter, consider the $s$ candidates in $T$ whose Borda score is the smallest. Now, take the sum of these scores, and average it over all the voters. Therefore, the $s$-Borda score assumes each voter is represented by the $s$ best candidates in $T$ according to her ranking, so that optimizing this score implies a form of proportional representation, where each voter on average has $s$ ``good'' candidates representing her. 

Our goal in this paper is to study the computational complexity (in $m, n$) of finding good committees according to the $s$-Borda score function. In particular (though not exclusively), we focus on the analysis of greedy algorithms, which are appealing for their simplicity and ease of use, especially in settings involving human decision making, such as parliamentary elections or participatory budgeting with ordinal preferences. 


\subsection{Results for $1$-Borda Score (Chamberlin-Courant Rule)}
Our main results focus on the canonical case where $s = 1$. This case has been extensively studied in computational social choice~\cite{LuB,SkworonFL,SkworonFS15,Heuristics,Elkind17,BrillFST19}, starting with the work of Chamberlin and Courant~\cite{CC}.  Here, each voter $v$ is represented by candidate $\argmin_{c \in T} r_v(c)$, that is, the most preferred candidate from $T$ in $v$'s ordering. The score of the voter is the rank of its representative, and the goal is to minimize the average of this score over the voters. This rule is also called the {\em Chamberlin-Courant} voting rule, though we will henceforth call it the $1$-Borda score for consistency with the generalizations we study later.

The $1$-Borda score is an ordinal version of the celebrated $k$-medians problem~\cite{AryaGKMMP04}. Unfortunately, for the ordinal version, it is not possible to approximate the minimum  score, $\opt$, to any constant factor in polynomial time unless $\p = \np$~\cite{SkworonFS15}.

\paragraph{The \g{} Algorithm.} A natural algorithm for the $1$-Borda score is the \g{} algorithm that iteratively adds the candidate that decreases the $1$-Borda score the most. This algorithm was analyzed in~\cite{LuB} as follows. Consider the maximization (or satisfaction) version where the score of a candidate $c$ for voter $v$ is $m + 1 - r_v(c)$, so that the score of a committee $T$ is $m + 1 - r_{\V}(T)$. Clearly, the maximization and minimization versions have the same optimum solutions, though they are very different from an approximation perspective. It is easy to check that the maximization objective is submodular~\cite{LuB}, so that \g{} is a $\left(1 - \frac{1}{e}\right)$-approximation by the classic result of~\cite{NWF}. However, this analysis only shows that \g{} yields a solution of score at most $m/e$ for the minimization objective.

Our first, and technically most challenging, contribution is an almost-tight analysis of this \g{} heuristic for the minimization version. In Section~\ref{sec:greedy}, we show that it achieves a score (given by Eq.~(\ref{eq:sborda})) of at most $2 \cdot\frac{m+1}{k+1}$ for any instance with $m$ candidates from which we need to choose a committee of size $k$. We complement this analysis by exhibiting an instance where \g{} has score at least $1.962 \cdot\frac{m+1}{k+1}$. 

For the maximization (or satisfaction) version, since the maximum possible score is $m$, the above result directly implies the following theorem.
\begin{theorem}
\label{thm:greedymax}
\g{} is a $\left(1 - \frac{2}{k+1}\right)$-approximation for the maximization version of $1$-Borda score.
\end{theorem} 
For $k$ larger than a small constant, this {\em significantly improves} the submodularity-based analysis~\cite{LuB} that only yields a $(1-1/e)$-approximation. Furthermore, it also improves on the best known approximation algorithm for this problem, Algorithm P in~\cite{SkworonFS15}, which achieves an approximation factor of $\left(1 - O\left(\frac{\ln k}{k}\right)\right)$.

At a technical level, the standard analysis of \g{} for maximizing submodular functions shows that the next candidate yields an improvement in objective that is at least $1/k$ fraction of the gap between the current solution and the optimum. We use Cauchy-Schwarz inequality on per-voter improvements to show an overall improvement per step that has a {\em quadratic} dependence on the gap. This yields a significant improvement when the gap is large, and is the crux of why we are able to improve the upper-bound analysis of the maximization version significantly. Our lower bound instance works by carefully choosing per-voter improvements that make Cauchy-Schwarz inequality almost tight. This requires a non-trivial construction where the candidates chosen by \g{} have ranks that lie on a carefully chosen spiral, and these are interspersed with candidates for whom voters' preferences are random. The ranks in each subsequent layer of the spiral decrease by a factor equal to the golden ratio.

\paragraph{A Benchmark and an Optimal Algorithm.} The next natural question we ask is: How much better can we do in polynomial time? In Section~\ref{sec:lower}, we show some hardness results for the minimization version (Eq~(\ref{eq:sborda})). 
Our main result {\em significantly improves} the constant factor hardness of approximation result of~\cite{SkworonFS15} and shows the following: 

\begin{theorem}
\label{thm:main_lb}
\label{thm:hardness}
Unless $\zpp = \np$, no polynomial-time algorithm can distinguish between instances with $\opt \ge (1 - o(1)) \cdot \frac{m+1}{k+1}$ from those with either:
\begin{enumerate}
    \item $\opt \le \left(\frac{m+1}{k+1}\right)^{\delta}$, where $\delta \in (0,1)$ is a constant; or
    \item $\opt \le \frac{1}{k^{\alpha}} \cdot \frac{m+1}{k+1}$, where $\alpha > 0$ is a constant.
\end{enumerate}
\end{theorem}
This construction yielding this theorem is delicate. We require the full power of Feige's hardness proof of {\sc Max Cover}~\cite{Feige}, in particular, that it works on ``regular'' instances where each set has the same size, and where a collection of disjoint sets cover the instance completely in the ``YES'' case.

Theorem~\ref{thm:main_lb} motivates us to define the score $\frac{m+1}{k+1}$ as a reasonable benchmark for this problem, and we call any efficient algorithm achieving this score as an ``optimal algorithm''. Such a benchmark is appealing in that it helps us analyze other simple and natural algorithms that have been proposed in literature, and perform a more fine-grained comparison. As we have already seen, the \g{} algorithm is always within a factor of $2$ of this benchmark. 

We now observe that if we pick a subset of $k$ candidates at random from $\C$, the expected score is exactly the benchmark $\frac{m+1}{k+1}$. We therefore denote the score $\frac{m+1}{k+1}$ as $\rand$. Now, we can design a deterministic optimal algorithm via derandomizing this randomized algorithm. Interestingly, we show that this derandomization yields a greedy algorithm that is exactly the same as the \b{} algorithm proposed in~\cite{Heuristics} as a polynomial time heuristic for this problem. In that work, the \b{} algorithm was derived by viewing the problem as a cooperative game where players are candidates, and coalitions are committees, and adapting the notion of Banzhaf score of coalitions~\cite{Banzhaf,dubeyS}. It was emprically shown to be a very effective heuristic for this problem, beating \g{} on most instances. We justify this empirical observation by viewing the \b{} algorithm instead as a derandomization of an optimal randomized algorithm. 

In summary, we show the following theorem. 
\begin{theorem}
\label{thm:main_banzhaf}
The \b{} algorithm achieves a minimization objective of at most $\rand = \frac{m+1}{k+1}$ in polynomial time, and is a $\left(1 - \frac{1}{k+1}\right)$-approximation to the maximization objective of $1$-Borda.
\end{theorem}
To complete the picture, we show that an easy consequence of Theorem~\ref{thm:main_lb} is that the approximation factor of $\left(1 - \frac{1}{k+1}\right)$ is best possible for the maximization version unless $\np = \zpp$.

\paragraph{Committee Monotonicity.} One appealing property of \g{} is that it is \emph{committee-monotone}~\cite{Elkind17}: The committee found for a smaller $k$ is always a subset of a committee found for larger $k$'s. This is immediate because \g{} adds the next candidate to the committee based on the improvement in the $1$-Borda score, and this improvement does not depend on $k$. On the other hand, the \b{} algorithm requires knowledge of $k$ at each greedy step, and is therefore not committee-monotone. We therefore ask: Is there a committee-monotone algorithm that can achieve the benchmark $\rand$? In Section~\ref{sec:monotone}, we answer this question in the negative: There exist instances where any committee-monotone algorithm has score at least $1.015 \cdot \rand$. This shows a separation between committee-monotone algorithms and an optimal algorithm such as the \b{} algorithm.

\paragraph{Connection to the Core.} The notion of core from cooperative game theory is appealing as a notion of fairness, and provides a strong notion of proportionality. Informally, in a core solution, every reasonably large subgroup of voters is happy in the sense that they do not all prefer the same candidate outside the chosen committee. Formally, the work of~\cite{JiangMW20,ChengJMW20} defines an $\alpha$-approximate core as follows. Fix some $\alpha \ge 1$. Given a committee $T$ of size $k$, a candidate $c$ is blocking if at least $\alpha \cdot \frac{n}{k}$ voters prefer $c$ to any candidate in $T$, that is,
\begin{equation}
    \label{eq:core}
\left| \Big\{v \in \V \ \Big| \ r_v(c) < \min_{c' \in T} r_v(c') \Big\} \right| \ge \alpha \cdot \frac{n}{k}.
\end{equation}
A committee $T$ is in the $\alpha$-approximate core if it does not admit  a blocking candidate. The work of~\cite{ChengJMW20,JiangMW20} shows that a $16$-approximate core always exists and can be computed in polynomial time, while a $(2-\varepsilon)$-approximate core is not guaranteed to exist for any constant $\varepsilon > 0$.

In Section~\ref{sec:core}, we show that the core indeed achieves a stronger notion of proportionality than the $1$-Borda score in the following sense: Any $\alpha$-approximate core solution has $1$-Borda score at most $\alpha (1+1/k) \cdot \rand$. The converse of this statement is however false: None of \opt{}, \g{}, or \b{} lies in an $\alpha$-approximate core for any constant $\alpha$. 

\subsection{Results for $s$-Borda Score}
We next consider the $s$-Borda score for $1 < s \le k$. We start with an analysis of the natural extensions to the greedy algorithms considered above for $s = 1$. It is easy to show that choosing a random committee of size $k$ yields expected score $\frac{s(s+1)}{2} \cdot \rand$, where as before, $\rand = \frac{m+1}{k+1}$. This implies its derandomization -- the \b{} algorithm -- has score at most $\frac{s(s+1)}{2} \cdot \rand$. Furthermore, there are instances where the best possible score $\opt \ge \frac{s(s+1)}{2} \cdot \rand$.

\paragraph{Analysis of \g{}.}
The \g{} algorithm extends naturally to this setting. In Section~\ref{sec:sborda_greedy} and Appendix~\ref{app:sborda}, we extend the result in Section~\ref{sec:greedy} to show that \g{} achieves an $s$-Borda score of at most $2s^2 \cdot \rand$, which is within a factor of $\frac{4s}{s + 1}$ of the upper bound for the \b{} algorithm.  

For the maximization version, recall that the score of candidate $c$ for voter $v$ is $m + 1 - r_c(v)$, and the voter's score for a committee is the sum of top $s$ candidate scores. Since the maximum possible score at most $ms$, this directly implies the following theorem.  For $s = o(k)$, this again {\em significantly improves} on the classic submodularity-based analysis that only shows a $\left(1 - \frac{1}{e}\right)$-approximation.
\begin{theorem}
\g{} is a  $\left(1-\frac{2s}{k+1}\right)$-approximation for the maximization version of $s$-Borda.
\end{theorem}
Note that for the related maximum multi-cover problem~\cite{Barman}, the approximation factor of $(1-1/e)$ is actually tight for \g{}. Therefore, our analysis of \g{} points to fundamental algorithmic differences between {\sc Max Multi-Cover} (resp. {\sc Max Cover}~\cite{Feige}) and $s$-Borda (resp. $1$-Borda), since we obtain significantly better factors for the latter.

\paragraph{Improved Algorithm.}
In contrast with the $s=1$ case, for larger values of $s$, we can obtain a non-trivial improvement over these greedy algorithms (for the minimization version of $s$-Borda) by using the natural LP relaxation for this problem~\cite{Byrka}. In Section~\ref{sec:lp}, we devise a randomized algorithm that is based on carefully combining dependent rounding of this LP solution with choosing a committee by uniform random sampling. We show that this  algorithm achieves expected score 
\[
\alg \le 3 \cdot \opt + O\left(s^{3/2} \log s \right) \cdot \rand,
\]
where $\opt = \min_{T \subseteq \C, |T| = k} r_{\V}(T)$ is the minimum possible $s$-Borda score of any committee. 

This result improves on the aforementioned bounds for greedy algorithms when $\opt$ is small. For instance, if $\opt = O(s^{3/2}) \cdot \rand$, the improvement is $\tilde{\Omega}(\sqrt{s})$.  We note that such an improved bound cannot be achieved by \g{} when $\opt$ is small: In Section~\ref{sec:sborda_g_lb}, we show instances where $\opt = o(1) \cdot \rand$, while the score of \g{} is $\Omega(s^2)\cdot \rand$. On the flip side, our improved bound is based on solving and rounding an LP, and is therefore not as simple or intuitive as the \g{} or \b{} algorithms.


\subsection{Related Work}
The literature on multiwinner elections is too vast to survey here. We present a survey of computational results in this space to place our work in context.

Suppose the candidates and voters are embedded in a metric space, and suppose  the score $r_v(c)$ is not the Borda score, but instead the metric distance between voter $v$ and candidate $c$. Then the objective for the $s = 1$ case is precisely the celebrated $k$-medians objective~\cite{AryaGKMMP04,JainV,LinV,CharikarGTS}, while the general-$s$ case has been studied as fault-tolerant $k$-medians~\cite{fault}. For both these problems, constant-factor approximation algorithms are known. The versions we consider can therefore be viewed as ordinal versions of $k$-medians and fault tolerant $k$-medians respectively. Towards showing better bounds for the ordinal versions, it is tempting to impose a condition such as the ordinal preferences of voters should correspond to distances in some underlying metric space. However, it is easy to show that given any set of ordinal preferences, there is a metric space that can realize these preferences, which means this assumption does not help. Nevertheless, the LP relaxation we use to derive improved bounds for the $s$-Borda score is the same as the standard LP relaxation for the (fault-tolerant) $k$-medians objective~\cite{fault}. It is an interesting open question to explore what other natural assumptions on voter preferences will lead to improved upper bounds for the $1$-Borda and $s$-Borda objectives.

The work of~\cite{Elkind17} considers generalizations of the $1$-Borda and $s$-Borda scores to \emph{committee scoring rules}. A committee scoring rule is a function that for each voter $v$ and committee $T$ of size $k$, maps the set of ordinal ranks $\{r_v(c), c \in T\}$ to a score. The $s$-Borda score we consider is an example of a decomposable rule, meaning that the score can be written as a sum of contributions from the committee members. The work of~\cite{SkworonFL,SkworonFS15} defines a special case of committee scoring rules where the score is a weighted sum of ranks of the committee members. They call these Ordered Weighted Average (OWA) rules. Again, it is easy to see that the $s$-Borda rule is an OWA rule. For the maximization version of committee scoring rules, the \g{} algorithm continues to be a $\left(1 - \frac{1}{e}\right)$-approximation via submodularity. It is an interesting open question to extend the results in this paper to other rules that achieve diverse or proportional representation~\cite{BrillFST19,AzizB20}.

Finally, the work of~\cite{Byrka} considers the variant where the score $r_v(c)$ is an arbitrary cardinal value, which is different from our focus on ordinal preferences. They consider the ``harmonic'' OWA rule where a voter assigns weight $1$ to candidate in the committee with lowest score, weight $\frac{1}{2}$ to the candidate with second lowest score, and so on, till weight $\frac{1}{k}$ to the candidate with highest score. They show this version has a constant-factor approximation algorithm by randomized rounding of the natural LP relaxation, first used in~\cite{cornuejols1983}. The difficulty with the $s$-Borda rule is that the weight jumps discretely from $1$ to $0$ when we move from the top $s$ candidates for a voter to the $(s+1)^{\text{st}}$ candidate. This discontinuity is most pronounced for $s = 1$, and leads to our strong impossibility result. In essence, this discontinuity is what motivates us to consider an alternate benchmark to analyse the performance of natural greedy algorithms assuming ordinal preferences.

\section{Preliminaries}
\label{secLprelim}

We consider the problem of selecting a subset of cardinality $k$ from a set $\C$ of $m$ candidates. We call this subset a \emph{committee}. A set $\V$ of $n$ voters express their preferences on the candidates ordinally. Each voter $v$ has a bijective ranking function $r_v : \C \rightarrow \{1, 2, \ldots, m\}$, and $v$ prefers those $c$'s with smaller $r_v(c)$. For example, the top-ranked candidate of $v$, denoted by $c_{\mathrm{top}(v)}$, satisfies $r_v\left(c_{\mathrm{top}(v)}\right) = 1$, and the bottom-ranked $c_{\mathrm{bot}(v)}$ satisfies $r_v\left(c_{\mathrm{bot}(v)}\right) = m$.

In $s$-Borda score, the cost for a voter $v$ of a committee $T$ is the sum of her ranks of the top $s$ candidates in $T$: $r_v(T) = \min_{Q \subseteq T, |Q| = s} \sum_{c \in Q} r_v(c)$. Further, the $s$-Borda score ($s \leq k$) of a committee $T$ is the average cost for all voters:
\[
r_{\V}(T) = \frac{1}{n} \sum_{v \in \V} r_v(T) = \frac{1}{n} \sum_{v \in \V} \left(\min_{Q \subseteq T, |Q| = s} \sum_{c \in Q} r_v(c)\right).
\]
In particular, when $s = 1$, $r_{\V}(T) = \frac{1}{n} \sum_{v \in \V} \min_{c \in T} r_v(c)$.

Fix a voter $v$ and look at her ranking on $\C$. If we pick a random size-$k$ subset of $\C$, the $t^{\text{th}}$ smallest rank is $t \cdot \frac{m + 1}{k + 1}$ in expectation. (See Appendix~\ref{app:folklore} for a proof of this well-known fact.)  
 Therefore, the expected performance of a random committee is
\[
\E_{T \subseteq \C}[r_{\V}(T)] = \frac{1}{n} \sum_{v \in \V} \E_{T \subseteq \C}[r_{v}(T)] = \sum_{t = 1}^s t \cdot \frac{m + 1}{k + 1} = \frac{s(s + 1)}{2} \cdot \frac{m + 1}{k + 1}.
\]
Define the benchmark $\rand$ to be the expected performance of a random committee when $s = 1$ as $\rand = \frac{m + 1}{k + 1}$. We will justify this benchmark in the subsequent sections.

We consider two simple committee-selection rules: \g{} and \b{}. These algorithms run in $k$ iterations, during which they build sets $\varnothing = T_0 \subsetneq T_1 \subsetneq \cdots \subsetneq T_k$, and declare $T_k$ as the selected committee.

In the $j^{\text{th}}$ iteration, \g{} picks candidate $c_j \in \C \setminus T_{j - 1}$ that minimizes $r_{\V}(T_{j - 1} \cup \{c_j\})$, and let $T_j = T_{j - 1} \cup \{c_j\}$. \b{}~\cite{Banzhaf,Heuristics}, on the other hand, picks candidate $c_j \in \C \setminus T_{j - 1}$ to minimize
\[
\sum_{\substack{S \subseteq \C: |S| = k \\ S \supseteq T_{j - 1} \cup \{c_j\}}} r_{\V}(S)
\]
in the $j^{\text{th}}$ iteration, and then sets $T_j = T_{j - 1} \cup \{c_j\}$. In other words, it greedily picks the candidate that minimizes the final score if the rest of the committee is chosen uniformly at random. Both \g{} and \b{} can run in polynomial time~\cite{Heuristics}.

Throughout the paper, we use \rand{}, \g{} and \b{} to denote either the algorithms or their performances, which should be clear from the context.

\section{Analysis of \g{} for $1$-Borda}
\label{sec:greedy}
In this section, we analyze the performance of \g{}, evaluated with respect to the benchmark \rand{}. Throughout this section, we only consider the $1$-Borda score, i.e., $s = 1$. We first show an upper bound that $\g \leq 2 \cdot \rand$, and then present an almost-matching lower-bound instance where $\g > 1.962 \cdot \rand$.

\subsection{Upper Bound}
\label{sec:greedy_ub}
Now we show $\g \leq 2 \cdot \rand$ as an upper bound. We first present the following lemma, which gives a lower bound on the improvement at each iteration.
\begin{lemma}
Let $T_t$ and $T_{t + 1}$ be the set of candidates produced by \g{} in the $t^{\text{th}}$ and $(t + 1)^{\text{st}}$ iterations, and $r_{\V}(T_t)$, $r_{\V}(T_{t + 1})$ be their respective score. We have:
\[
r_{\V}(T_t) - r_{\V}(T_{t + 1}) \geq \frac{\sum_{v \in \V} r_v(T_t)(r_v(T_t) - 1)}{2n(m - t)}.
\]
\label{lem:minimum_improvement}
\end{lemma}
\begin{proof}
For a candidate $c \notin T_t$, define $\Delta_c := r_{\V}(T_t) - r_{\V}(T_t \cup \{c\})$, i.e., the current marginal contribution of $c$ to the $1$-Borda score. Taking the sum of $\Delta_c$ over $c \notin T_t$:
\[
\sum_{c \in \C \setminus T_t} \Delta_c = \frac{1}{n} \sum_{v \in \V} \sum_{j = 1}^{r_v(T_t) - 1} j = \frac{\sum_{v \in \V} r_v(T_t)(r_v(T_t) - 1)}{2n}.
\]

\g{} chooses $c^* = \argmax_c \Delta_c$ at the $(t + 1)^{\text{st}}$ iteration, giving us
\[
r_{\V}(T_t) - r_{\V}(T_{t + 1}) = \Delta_{c^*} \geq \frac{1}{m - t} \sum_{c \in \C \setminus T_t} \Delta_c = \frac{\sum_{v \in \V} r_v(T_t)(r_v(T_t) - 1)}{2n(m - t)}. \qedhere
\]
\end{proof}

Now we prove our upper bound of $2$.
\begin{theorem}
$\g \leq 2 \cdot \rand$.
\label{thm:ub_greedy_1}
\end{theorem}
\begin{proof}
We prove by induction. As the base case where $k = 1$, $\g \leq m < m + 1 =  2 \cdot \rand$.
Now suppose that the claim holds for some $k - 1$ and we will prove that it also holds for $k$. By induction hypothesis, we have:
\[
r_{\V}(T_{k - 1}) \leq 2 \cdot \frac{m + 1}{k}.
\]

If $r_{\V}(T_{k - 1}) \leq 2 \cdot \frac{m + 1}{k + 1}$, then $r_{\V}(T_k) \leq r_{\V}(T_{k - 1}) \leq 2 \cdot \frac{m + 1}{k + 1}$ finishes the proof. Thus, we only need to consider the following case:
\[
2 \cdot \frac{m + 1}{k + 1} < r_{\V}(T_{k - 1}) \leq 2 \cdot \frac{m + 1}{k}.
\]

We now have the following, where the first inequality is by Lemma~\ref{lem:minimum_improvement} and second by Cauchy-Schwarz inequality:
\begin{align*}
r_{\V}(T_{k - 1}) - r_{\V}(T_k) &\geq \frac{\sum_{v \in \V} r_v(T_{k - 1})(r_v(T_{k - 1}) - 1)}{2n(m - k + 1)} \\
&\geq \frac{\frac{1}{n}(\sum_{v \in \V}r_v(T_{k - 1}))^2 - \sum_{v \in \V}r_v(T_{k - 1})}{2n(m - k + 1)} \\
&= \frac{(\sum_{v \in \V}r_v(T_{k - 1}))^2}{2n^2(m + 1)} \cdot \frac{m + 1}{m - k + 1} \cdot \frac{\sum_{v \in \V}(r_v(T_{k - 1}) - 1)}{\sum_{v \in \V}r_v(T_{k - 1})}.
\end{align*}

Since $r_{\V}(T_{k - 1}) \geq 2 \cdot \frac{m + 1}{k + 1}$ by assumption, we have:
\begin{align*}
\frac{m + 1}{m - k + 1} \cdot \frac{\sum_{v \in \V}(r_v(T_{k - 1}) - 1)}{\sum_{v \in \V}r_v(T_{k - 1})} &\geq \frac{m + 1}{m - k + 1} \cdot \frac{2 \cdot \frac{m + 1}{k + 1} - 1}{2 \cdot \frac{m + 1}{k + 1}}\\
&= \frac{2(m + 1) - k - 1}{2(m + 1) - 2k} \geq 1.
\end{align*}
Combining the previous two inequalities, we therefore have:
\[
r_{\V}(T_{k - 1}) - r_{\V}(T_k) \geq \frac{(\sum_{v \in \V}r_v(T_{k - 1}))^2}{2n^2(m + 1)} = \frac{r_{\V}^2(T_{k - 1})}{2(m + 1)},
\]
which is equivalent to:
\[
r_{\V}(T_k) \leq - \frac{1}{2(m + 1)} r_{\V}^2(T_{k - 1}) + r_{\V}(T_{k - 1}).
\]

Notice that the right hand side is a quadratic function in $r_{\V}(T_{k - 1})$, which is monotonically increasing for $r_{\V}(T_{k - 1}) \leq m + 1$. Since $r_{\V}(T_{k - 1}) \leq 2 \cdot \frac{m + 1}{k} \leq m + 1$, the right hand side reaches its maximum at $2 \cdot \frac{m + 1}{k}$. Thus, we have:
\[
r_{\V}(T_k) \leq - \frac{1}{2(m + 1)}\cdot \left(\frac{2(m + 1)}{k}\right)^2 + \frac{2(m + 1)}{k} \leq \frac{2(m + 1)}{k + 1},
\]
which concludes our induction.
\end{proof}

\paragraph{Proof of Theorem~\ref{thm:greedymax}.} For the maximization version, the above result implies \g{} achieves score at least $(m+1) \cdot \left(1- \frac{2}{k+1} \right)$. Since the maximum possible score is $m$, this implies that \g{} is a $\left( 1 - \frac{2}{k+1} \right)$-approximation.

\subsection{Lower Bound}
\label{sec:greedy_lb}
Now we complement our result with a lower-bound example for \g.

\begin{theorem}
\label{thm:greedy_lb}
There exists an instance in which $r_{\V}(T_k) > 1.962\cdot \rand$.
\end{theorem}

\paragraph{Construction.} In the sequel, we will prove the above theorem. In the instance we construct, $m$, $n$, and $k$ are all sufficiently large. For convenience of illustration, we scale down the ranks by a factor of $m$: now the ranks are $\frac{1}{m}, \frac{2}{m}, \ldots, \frac{m - 1}{m}, 1$. As $m \to \infty$, $\frac{1}{m} \to 0$, so the set of ranking $\{\frac{1}{m}, \frac{2}{m}, \ldots, 1\}$ will become dense in $[0, 1]$, and thus we regard the ranking as being continuous from $0$ to $1$. Our goal becomes to construct an instance in which \g{} gives $r_{\V}(T_k) > 1.962 \cdot \frac{1}{k + 1}$. 

There are sufficiently many voters, enabling us to view them as a continuum from $0$ to $1$, forming a circle (the base in Fig.~\ref{fig:cylinder}) with angular position ranging from $0$ to $2 \pi$. Imagine that each voter writes down her favorite, her second favorite, \ldots, her least favorite candidate in that order vertically. The result is the side of a cylinder with height $1$, as depicted in Fig.~\ref{fig:cylinder}. Each point on the side identifies a candidate, whose distance to the top, $d$, indicates the corresponding voter ranks him as her $(dm)^{\text{th}}$ favorite candidate (i.e., the candidate has a rank of $d$ in the voter's preference after scaling).


\begin{figure}[H]
\centering
\begin{tikzpicture}[scale = 1.0, yscale = 2.5, xscale = 4.0]
\draw[domain=-1.002:1.002, variable=\x, samples=1000, smooth, very thick, dotted] plot ({\x}, {0.1+0.2*((1-\x*\x)^2)^0.25});
\draw[domain=-1.002:1.002, variable=\x, samples=1000, smooth, very thick] plot ({\x}, {0.1-0.2*((1-\x*\x)^2)^0.25});
\draw[domain=-1.002:1.002, variable=\x, samples=1000, smooth, ultra thick, color3] plot ({\x}, {2+0.2*((1-\x*\x)^2)^0.25});
\draw[domain=-1.002:1.002, variable=\x, samples=1000, smooth, ultra thick, color3] plot ({\x}, {2-0.2*((1-\x*\x)^2)^0.25});

\draw[domain=0.1:2, variable=\y, samples=20, smooth, very thick] plot ({-1}, {\y});
\draw[domain=0.1:2, variable=\y, samples=20, smooth, very thick] plot ({1}, {\y});

\draw[domain=-1:1, variable=\x, color1, samples=200, smooth, very thick] plot ({\x}, {2*(1-0.4*exp(-0.076587*(asin(\x)/180*pi+pi/2))) - 0.2*(1-\x*\x)^0.5});
\draw[domain=-1:1, variable=\x, color1, samples=200, smooth, very thick, dotted] plot ({\x}, {2*(1-0.4*((5^0.5-1)/2)^0.5*exp(-0.076587*(asin(-\x)/180*pi+pi/2))) + 0.2*(1-\x*\x)^0.5});

\draw[domain=-1:1, variable=\x, color1, samples=200, smooth, very thick] plot ({\x}, {2*(1-0.4*((5^0.5-1)/2)^1.0*exp(-0.076587*(asin(\x)/180*pi+pi/2))) - 0.2*(1-\x*\x)^0.5});
\draw[domain=-1:1, variable=\x, color1, samples=200, smooth, very thick, dotted] plot ({\x}, {2*(1-0.4*((5^0.5-1)/2)^1.5*exp(-0.076587*(asin(-\x)/180*pi+pi/2))) + 0.2*(1-\x*\x)^0.5});

\draw[domain=-1:1, variable=\x, color1, samples=200, smooth, very thick] plot ({\x}, {2*(1-0.4*((5^0.5-1)/2)^2.0*exp(-0.076587*(asin(\x)/180*pi+pi/2))) - 0.2*(1-\x*\x)^0.5});
\draw[domain=-1:1, variable=\x, color1, samples=200, smooth, very thick, dotted] plot ({\x}, {2*(1-0.4*((5^0.5-1)/2)^2.5*exp(-0.076587*(asin(-\x)/180*pi+pi/2))) + 0.2*(1-\x*\x)^0.5});

\draw[domain=-1:1, variable=\x, color1, samples=200, smooth, very thick] plot ({\x}, {2*(1-0.4*((5^0.5-1)/2)^3.0*exp(-0.076587*(asin(\x)/180*pi+pi/2))) - 0.2*(1-\x*\x)^0.5});
\draw[domain=-1:1, variable=\x, color1, samples=200, smooth, very thick, dotted] plot ({\x}, {2*(1-0.4*((5^0.5-1)/2)^3.5*exp(-0.076587*(asin(-\x)/180*pi+pi/2))) + 0.2*(1-\x*\x)^0.5});

\draw[domain=-1:1, variable=\x, color1, samples=200, smooth, very thick] plot ({\x}, {2*(1-0.4*((5^0.5-1)/2)^4.0*exp(-0.076587*(asin(\x)/180*pi+pi/2))) - 0.2*(1-\x*\x)^0.5});
\draw[domain=-1:1, variable=\x, color1, samples=200, smooth, very thick, dotted] plot ({\x}, {2*(1-0.4*((5^0.5-1)/2)^4.5*exp(-0.076587*(asin(-\x)/180*pi+pi/2))) + 0.2*(1-\x*\x)^0.5});

\draw[domain=-0.29:0.29, variable=\x, color2, samples=200, smooth, ultra thick] plot ({\x}, {2*(1-0.4*exp(-0.076587*(asin(\x)/180*pi+pi/2))) - 0.2*(1-\x*\x)^0.5});

\draw[domain=-0.02:0.02, variable=\y, color2, samples=20, smooth, very thick] plot ({-0.29}, {\y+2*(1-0.4*exp(-0.076587*(asin(-0.29)/180*pi+pi/2))) - 0.2*(1-(-0.29)*(-0.29))^0.5});

\draw[domain=-0.02:0.02, variable=\y, color2, samples=20, smooth, very thick] plot ({-0.10}, {\y+2*(1-0.4*exp(-0.076587*(asin(-0.10)/180*pi+pi/2))) - 0.2*(1-(-0.10)*(-0.10))^0.5});

\draw[domain=-0.02:0.02, variable=\y, color2, samples=20, smooth, very thick] plot ({0.10}, {\y+2*(1-0.4*exp(-0.076587*(asin(0.10)/180*pi+pi/2))) - 0.2*(1-(0.10)*(0.10))^0.5});

\draw[domain=-0.02:0.02, variable=\y, color2, samples=20, smooth, very thick] plot ({0.29}, {\y+2*(1-0.4*exp(-0.076587*(asin(0.29)/180*pi+pi/2))) - 0.2*(1-(0.29)*(0.29))^0.5});

\draw[domain=-0.188:0.188, variable=\x, color2, samples=200, smooth, ultra thick] plot ({\x}, {2*(1-0.4*((5^0.5-1)/2)^2.0*exp(-0.076587*(asin(\x)/180*pi+pi/2))) - 0.2*(1-\x*\x)^0.5});

\draw[domain=-0.02:0.02, variable=\y, color2, samples=20, smooth, very thick] plot ({-0.188}, {\y + 2*(1-0.4*((5^0.5-1)/2)^2.0*exp(-0.076587*(asin(-0.188)/180*pi+pi/2))) - 0.2*(1-(-0.188)*(-0.188))^0.5});

\draw[domain=-0.02:0.02, variable=\y, color2, samples=20, smooth, very thick] plot ({-0.114}, {\y + 2*(1-0.4*((5^0.5-1)/2)^2.0*exp(-0.076587*(asin(-0.114)/180*pi+pi/2))) - 0.2*(1-(-0.114)*(-0.114))^0.5});

\draw[domain=-0.02:0.02, variable=\y, color2, samples=20, smooth, very thick] plot ({-0.0382}, {\y + 2*(1-0.4*((5^0.5-1)/2)^2.0*exp(-0.076587*(asin(-0.0382)/180*pi+pi/2))) - 0.2*(1-(-0.0382)*(-0.0382))^0.5});

\draw[domain=-0.02:0.02, variable=\y, color2, samples=20, smooth, very thick] plot ({0.0382}, {\y + 2*(1-0.4*((5^0.5-1)/2)^2.0*exp(-0.076587*(asin(0.0382)/180*pi+pi/2))) - 0.2*(1-(0.0382)*(0.0382))^0.5});

\draw[domain=-0.02:0.02, variable=\y, color2, samples=20, smooth, very thick] plot ({0.114}, {\y + 2*(1-0.4*((5^0.5-1)/2)^2.0*exp(-0.076587*(asin(0.114)/180*pi+pi/2))) - 0.2*(1-(0.114)*(0.114))^0.5});

\draw[domain=-0.02:0.02, variable=\y, color2, samples=20, smooth, very thick] plot ({0.188}, {\y + 2*(1-0.4*((5^0.5-1)/2)^2.0*exp(-0.076587*(asin(0.188)/180*pi+pi/2))) - 0.2*(1-(0.188)*(0.188))^0.5});

\draw [very thick, ->] (1.2, 2) -- (1.2, 0.1);

\node [color3] at (-0.9, 2.2) {Voters};
\node [color1] at (-0.5, 1.25) {Critical Candidates};
\node [color2] at (0, 1) {Lower-Layer Candidates};
\node [color2] at (0, 1.45) {Higher-Layer Candidates};
\node at (1.45, 1.4) {Decreasing};
\node at (1.45, 1.28) {Preferences};

\node at (1.40, 2.0) {Rank: $0$};
\node at (1.40, 0.1) {Rank: $1$};

\draw [decorate, very thick, color2, decoration={brace,amplitude=5pt}] (-1, {2-0.8}) -- (-1, {2-0.8*(5^0.5-1)/2});
\draw [decorate, very thick, color2, decoration={brace,amplitude=5pt}] (-1, {2-0.8*(5^0.5-1)/2}) -- (-1, {2-0.8*((5^0.5-1)/2)^2});

\node [color2] at (-1.25, {(2-0.8+2-0.8*(5^0.5-1)/2)/2}) {$1^{\text{st}}$ Layer};
\node [color2] at (-1.25, {(2-0.8*(5^0.5-1)/2+2-0.8*((5^0.5-1)/2)^2)/2}) {$2^{\text{nd}}$ Layer};

\end{tikzpicture}
\caption{Construction of the Bad Instance for \g{}}
\label{fig:cylinder}
\end{figure}
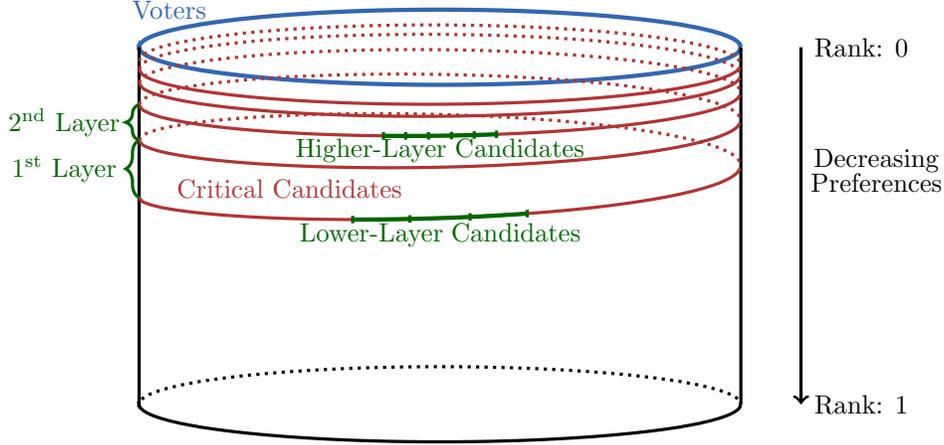

We divide the set of candidates into two types -- {\em critical} and {\em dummy}. The former set has size $k \ll m$, and the latter has size $m-k$. Our proof will show that \g{} will choose the critical candidates in a fixed order, and will not choose any dummy candidate.

The critical candidates are present in $\ell$ ``layers'' as shown in the red spiral in Fig.~\ref{fig:cylinder}, where $\ell$ is sufficiently large. This figure shows the ranks of the critical candidates in the voters' profiles. We parametrize this spiral by $\theta$, which maps to the voter at the corresponding angular position $2\pi\theta$. We place critical candidates in order, where each candidate appears a number of times consecutively on the spiral. Therefore, each voter has one critical candidate from each layer $t = 0, 1, \ldots, \ell$ in the spiral part of its ranking.

In the $t^{\text{th}}$ layer, the parameter $\theta$ lies in $[t - 1, t)$. The critical candidate when the parameter is $\theta$ has rank $g(\theta) = a \varphi^{\theta}$ for the voter at angular position $2\pi\theta$. Here, $\varphi$ denotes the golden ratio $\frac{\sqrt{5} - 1}{2} \approx 0.618$, and $a$ is a sufficiently small constant so that rounding to the nearest integer does not change the analysis. This critical candidate is placed for a certain length $h(\theta)$ on the spiral, which means this candidate appears at rank $g(\theta)$ for voters in the range $\left[2\pi\theta, 2\pi(\theta + h(\theta))\right]$. In our construction, $h(\theta)$ will be very small, so that we will say this candidate appears $h(\theta)$ times at rank $g(\theta)$ for parameter $\theta$. The greater $\theta$ is, the smaller $h(\theta)$ has to be, and we will calculate its expression later. 

For the convenience of analysis, at the layer $t = 0$, that is, for $\theta \in [-1,0)$, there is a special candidate appearing on the spiral throughout the layer. This special candidate is picked first by \g{}. Other than its appearance on the spiral, any critical candidate is placed at the very bottom, i.e., rank $1$, for the other voters. Denote the total number of critical candidates by $k$. Then we have $m - k$ dummy candidates. These dummy candidates are symmetrically placed at other ranks. We copy each voter $(m - k)!$ times, once for each possible permutation of the dummy candidates to place in the remaining ranks. 

The idea of this construction is to trick \g{} into picking every critical candidate on the spiral in order, while in fact, lower-layer critical candidates have no contribution to the objective once higher-layer ones have been selected. The following analysis computes the optimal parameters to realize this plan.

\paragraph{Not Choosing a Dummy Candidate.} We first ensure \g{} does not choose a dummy candidate in this instance by setting $h(\theta)$ properly. We assume that \g{} chooses critical candidates in increasing order of $\theta$, and we will justify this assumption later.

To simplify notation, denote $X = \int_{0}^1 a\varphi^{\theta}\d \theta$ and $Y = \int_{0}^1 a^2\varphi^{2\theta}\d \theta$. Computing these explicitly:
\[
X = \frac{a}{\ln\varphi}(\varphi - 1), \qquad Y = \frac{a^2}{2\ln\varphi}(\varphi^2 - 1) = X^2 \frac{(\varphi + 1)\ln\varphi}{2(\varphi - 1)}.
\]

Using this notation, consider the critical candidate at the beginning of the first layer, that is, at $\theta = 0$. Since \g{} chooses the candidate at layer $t = 0$, the decrease in score due to this critical candidate is:
\begin{equation}
\label{eq:improve_critical}
h(0) \cdot (g(-1) - g(0)) = h(0) \cdot a \cdot \left(\frac{1}{\varphi} - 1\right) = h(0) \cdot a \cdot \varphi.
\end{equation}
where we have used that since $\varphi$ is the golden ratio, $\varphi + \varphi^2 = 1$.

Now consider the dummy candidates. Just after \g{} has chosen the special candidate at layer $t = 0$, each such candidate improves the rank of $g(\theta-1)$ fraction of voters at $\theta \in [0, 1)$. This is because we placed all permutations of the dummy candidates with each voter $\theta$, and \g{} has already chosen the special candidate. By the same reasoning, conditioned on improvement, the average improvement is $g(\theta-1)/2$. Therefore, the decrease in score due to a dummy candidate is: 
\begin{equation}
\label{eq:improve_dummy}
\int_0^{1} \frac{g^2(\theta-1)}{2} \d \theta = \frac{a^2}{2\varphi^2} \int_0^{1} \varphi^{2\theta} \d \theta = \frac{1}{2\varphi^2} \cdot Y.
\end{equation}
Since we want \g{} to choose the critical candidate, we need to set
\[
h(0) = \frac{Y}{2\varphi^3 a}.
\]
By the symmetry of the spiral, an identical calculation now holds for all $\theta > 0$. To make \g{} choose the critical candidate at this location (assuming it has chosen critical candidates for smaller values of $\theta$), we need:
\[
h(\theta) = \frac{Y}{2\varphi^3 a} \varphi^{\theta}.
\]
Note that $h(\theta)$ depends linearly on $a$, so that for very small $a$, we can pretend this set of voters lies exactly at $\theta$. Further, $h(\theta)$ is decreasing with $\theta$.

\paragraph{Choosing Critical Candidates in Order.} We now show that \g{} chooses the critical candidates following the order on the spiral.
\begin{lemma}
\label{lem:greedyopt2}
\g{} chooses the critical candidates in increasing order of $\theta$.
\end{lemma}
\begin{proof}
The calculation is identical at any step of \g{}, so we focus on the step where \g{} is at the beginning of the first layer, that is, considering the critical candidate at $\theta = 0$. Recall that \g{} has chosen the special candidate at layer $t = 0$. The previous analysis showed that the critical candidate at $\theta = 0$ yields decrease of $\frac{Y}{2 \varphi^2}$. For critical candidates in the same layer $t = 1$ (that is, for $\theta \in [0, 1)$), the contribution of the candidate at $\theta$ is
\[
h(\theta) \cdot (g(\theta-1) - g(\theta)) = \frac{Y}{2\varphi^3} \varphi^{\theta} \left(\varphi^{\theta - 1} - \varphi^{\theta}\right) = \frac{Y}{2\varphi^2} \cdot \varphi^{2\theta},
\]
which decreases with $\theta$, so that the current candidate, $\theta = 0$, offers the best decrease. Here, we have used that since $\varphi$ is the golden ratio, $\varphi^2 + \varphi = 1$.

For $t \ge 1$, suppose we instead considered a candidate $t +\theta$ for $\theta \in [0,1)$ located in layer $t+1$. Conditioned on having chosen layer $t = 0$, this candidate gives a contribution of
\begin{align*}
h(t + \theta) \cdot (g(\theta-1) - g(t + \theta)) &\leq h(t) \cdot (g(-1) - g(t))\\
&\leq \max\big(h(2) \cdot g(-1), \ h(1) \cdot (g(-1) - g(1))\big)\\
&= \max\left(\frac{Y}{2\varphi^2}, \ \frac{Y}{2\varphi^2a} \cdot a\left(\frac{1}{\varphi} - \varphi\right)\right) = \frac{Y}{2\varphi^2},
\end{align*}
where the first inequality uses that $h(\theta)$ is decreasing in $\theta$, and that $\varphi < 1$.

Therefore, \g{} will pick the critical candidate at $\theta = 0$ instead of another candidate at the same or a higher layer. Since the argument is identical at each $\theta$, \g{} picks critical candidates in order on the spiral.
\end{proof}

\paragraph{The Lower Bound.} So far we have shown that \g{} chooses critical candidates in increasing order of layers and does not choose dummy candidates. We finally put it all together and show the following bound, which completes the proof of Theorem~\ref{thm:greedy_lb}. 

\begin{proof}[Proof of Theorem~\ref{thm:greedy_lb}]
The number of critical candidates on the $t^{\text{th}}$ layer ($\theta \in [t - 1, t)$) is
\[
\int_{t - 1}^t \frac{1}{h(\theta)} \d \theta = \frac{2 \varphi^3 a}{Y} \int_{t - 1}^t \varphi^{-\theta} \d \theta = \frac{2 a}{\varphi^{t-3} Y} \int_{-1}^0 \varphi^{-\theta} \d \theta  = \frac{2 a}{\varphi^{t-3} Y} \int_{0}^1 \varphi^{\theta} \d \theta = \frac{2 X}{\varphi^{t - 3} Y}.
\]
Therefore, when it is done with the $\ell^{\text{th}}$ layer, the number of candidates \g{} has picked is
\[
k = \frac{2 X}{\varphi^{\ell - 3} Y} (1 + \varphi + \varphi^2 + \cdots + \varphi^{\ell - 1}) \rightarrow \frac{2 X}{(1 - \varphi)\varphi^{\ell - 3} Y}
\]
when $\ell$ is large. Meanwhile, the $1$-Borda score of \g{} is
\[
r_{\V}(T_k) = \int_{0}^1 g(\ell - 1 + \theta) \d \theta = \varphi^{\ell - 1} X.
\]
Therefore, the approximation ratio is
\[
(k+1) r_{\V}(T_k) \ge \frac{2 X}{(1 - \varphi)\varphi^{\ell - 3} Y} \cdot \varphi^{\ell - 1} X = \frac{2 \varphi^2 X^2}{(1 - \varphi) Y} = \frac{2\varphi^2}{(1 - \varphi)} \cdot \frac{2(\varphi - 1)}{(\varphi + 1) \ln \varphi} = -\frac{4 \varphi^2}{(\varphi + 1) \ln \varphi} > 1.962. \qedhere
\]
\end{proof}

\section{Hardness of $1$-Borda and Optimal Deterministic Algorithm}
\label{sec:lower}
Throughout this section, we focus on $1$-Borda score. We justify our choice of benchmark $\rand = \frac{m + 1}{k + 1}$, and show that a deterministic algorithm, \b{}, achieves this benchmark optimally. First, notice that if the input consists of one voter for each possible preference of $m$ candidates (thus $n = m!$), picking any committee has the same $1$-Borda score by symmetry, so \opt{} is just \rand. Thus, we have the following proposition:
\begin{proposition}
For any $m$ and $k$, there exist instances where $\opt = \rand$.
\end{proposition}

\subsection{Hardness Result for $1$-Borda: Theorem~\ref{thm:main_lb}}
We now show Theorem~\ref{thm:hardness}: Even if $\opt$ is very small, it is computationally hard to significantly beat $\rand$. 
To prove this hardness result, we show a reduction from the decision version of the \rc{} problem. 
\begin{definition}
In \rc{}, these is a universe $U$ of $n$ elements $\{a_1, a_2, \ldots, a_n\}$, and a family $\mathcal{F} = \{S_1, S_2, \ldots, S_z\}$ of subsets of $U$. Each $S_i$ has the same size $\frac{n}{k}$. The {\em value} of an instance is the maximum size of the union of $k$ sets from $\mathcal{F}$. For any constant $\varepsilon > 0$, we consider the following decision version:
\begin{itemize}
    \item ``YES'' instances are those with value $n$. Therefore, there exist $k$ disjoint sets each of size $n/k$ that cover all the elements.
    \item ``NO'' instances are those with value at most $\frac{3}{4}n$.
\end{itemize}
\label{def:k_cover}
\end{definition}



The above problem known to be $\nph$ to approximate via the following lemma that is implicit in the proof of Theorem 5.3 in~\cite{Feige}.

\begin{lemma}[\cite{Feige}]
The decision version of \rc{} from Definition~\ref{def:k_cover} is $\nph$, that is, unless $\p = \np$, there is no polynomial time algorithm that can decide always answers ``YES'' for ``YES'' instances and answers ``NO'' for ``NO'' instances.
\label{lem:cited_hardness}
\end{lemma}

Note that if the instance has value $n$, there exist $k$ disjoint sets each of size $n/k$ that cover all the elements. This aspect will be crucial in our reduction. Also needed in our reduction, we state the following lemma for constructing a profile with polynomial number of voters, where the best solution with score $\opt$ has similar performance as $\rand$.

\begin{lemma}
Fix any $\varepsilon > 0$ and let $n \ge \left\lceil\frac{m(k + 1)^2}{\varepsilon^2}\right\rceil$. Consider the instance where the preference of each voter is an independent and uniformly random permutation. Let $\opt'$ denote the expected value of the optimum score, and $\rand' = \frac{m+1}{k+1}$, then $\Pr[\opt \leq (1 - \varepsilon) \cdot \rand] < \frac{1}{2}$, where the probability is over the randomness in the permutations.
\label{lem:random_construction}
\end{lemma}
\begin{proof}
Fix any committee $T$ of size $k$. Notice that $\E[r_{\V}(T)] = \rand'$ since the preferences are uniformly random. We have
\begin{align*}
\Pr[r_{\V}(T) - \rand' \leq -\varepsilon \cdot \rand'] &= \Pr\left[\frac{1}{n}\sum_{v \in \V} r_v(T) - \rand' > \varepsilon \cdot \rand'\right]\\
& \leq \exp\left(\frac{-2n (\varepsilon \cdot \rand')^2}{m^2}\right) \leq \exp\left(\frac{-2n \varepsilon^2}{(k + 1)^2}\right),
\end{align*}
where the second step comes from Hoeffding's inequality. By union bound,
\begin{align*}
\Pr[\opt' \leq (1 - \varepsilon) \cdot \rand'] &\leq \binom{m}{k} \cdot \Pr[r_{\V}(T) - \rand' \leq -\varepsilon \cdot \rand']\\
&\leq \exp\left(\frac{-2n \varepsilon^2}{(k + 1)^2} + m\right)\leq e^{-m} < \frac{1}{2}. \qedhere
\end{align*}
\end{proof}

Now we are ready to prove Theorem~\ref{thm:hardness}.
\begin{proof}[Proof of Theorem~\ref{thm:hardness}]
Fix a $\varepsilon > 0$ and let $\varepsilon' = 10 \varepsilon$.  We will choose $\varepsilon$ appropriately later. Given any instance of \rc{} with $n$ elements and $z$ sets each of size $n/k$ (as in Definition~\ref{def:k_cover}), we construct the following instance for our problem:
\begin{itemize}
    \item There are $N = n R$ voters $v_{ij}$ where $i \in [n]$ and $j \in [R]$.  We have $m = \frac{2}{\varepsilon'} k z$ candidates. The first $z$ candidates $\{c_1, c_2, \ldots, c_z\}$ are ``critical'' candidates, and the other $m-z$ candidates are ``dummy'' candidates. Each voter corresponds to an element in the universe and each critical candidate corresponds to a set in \rc{}. 
    \item If a set $S_i$ covers $a_j$, then voters $v_{ij}$ for $j \in [R]$ rank $c_j$ within top $\varepsilon'$ fraction. Otherwise, $v_{ij}$'s rank $c_j$ within bottom $\varepsilon'$ fraction.
    \item Independently for each voter, fill the rest of her preferences with the $m - z$ dummy candidates uniformly randomly. 
    \item  The copies of a voter only differ in the ranking of the dummy candidates. We set the number of copies to be $R = \left\lceil \frac{10 m k^2}{n \varepsilon^2} \right\rceil$. These copies are there to ensure Lemma~\ref{lem:random_construction} applies to the dummy candidates.
\end{itemize}

Clearly, the above construction has size $\mbox{poly}(1/\varepsilon,n,z,k)$. Let $\opt$ denote the optimal score on this instance. Recall that $\rand = \frac{m+1}{k+1}$. First suppose the instance of \rc{} has value $n$ (``YES'' instance) so that there are $k$ sets that cover all $n$ elements, then it is easy to check that choosing the corresponding critical candidates as the committee yields $\opt \le z < \varepsilon' \cdot \rand$. 

On the other hand, suppose the instance of \rc{} is such that any collection of sets of size $k$ only covers at most $(1-1/e + \varepsilon)n \le \frac{3}{4} n$ elements (``NO'' instance). Consider any committee $T$ and suppose $T = R \cup D$ where $R$ is a subset of critical candidates and $D$ is a subset of dummy candidates. Let $r = |R|$ and $d = |D| = k - r$. Let $n'$ be the number of elements $R$ covers in the $\rc{}$ instance. By assumption, $n - n' \ge \frac{n}{4}$ since any collection $R$ covers at most $\frac{3}{4} n$ elements. Further, since the instance is regular, $n - n' \le \frac{n}{k} r$, so that $n - n' \ge d \frac{n}{k}$.

Using Lemma~\ref{lem:random_construction}, with probability $> \frac{1}{2}$ over the choice of the ranking of the dummy candidates, the optimal score of $D$ on the $(n-n')R$ uncovered voters using the $m-z$ dummy candidates is greater than $(1-\varepsilon) \frac{m - z}{d + 1}$. Inserting the critical candidates cannot decrease this score for these voters, since the candidates in $R$ appear last in their ordering. Further, we have assumed $m = \frac{2}{\varepsilon'} k z$. Therefore, with probability $> \frac{1}{2}$, we have:
    $$ \opt > \frac{n - n'}{n} (1-\varepsilon) \frac{m - z}{d + 1} \ge \frac{n - n'}{n} \cdot \left(1 - \frac{\varepsilon'}{4}\right) \cdot \frac{m + 1}{d + 1}$$

We now split the analysis into two cases:

\begin{enumerate}
    \item Suppose $d + 1 \le \frac{k+1}{4}$. Since $n - n' \ge \frac{n}{4}$, we have 
    \[
    \opt > \frac{1}{4} \cdot \left(1 - \frac{\varepsilon'}{4}\right) \cdot \frac{m + 1}{d + 1} \ge \left(1 - \frac{\varepsilon'}{4}\right) \cdot \frac{m + 1}{k + 1} \ge  (1-\varepsilon') \rand.
    \]
    \item Suppose $d + 1 \ge \frac{k+1}{4}$. Since $n - n' \ge \frac{d}{k} n$, and since $d,k = \omega(1)$, we have:
    \[
    \opt > \frac{d}{k} \cdot \left(1 - \frac{\varepsilon'}{4}\right) \cdot \frac{m + 1}{d+1} = \frac{d}{d+1}\frac{k+1}{k} \left(1 - \frac{\varepsilon'}{4}\right) \cdot \frac{m + 1}{k+1} \ge (1-\varepsilon') \rand.
    \]
\end{enumerate}

Therefore, our construction ensures that with probability $> \frac{1}{2}$, we have $\opt \ge (1-\varepsilon') \rand$ if the original \rc{} instance has value at most $\frac{3}{4} n$.

Now suppose there is a polynomial time algorithm that can distinguish between instances with $\opt \le \varepsilon' \rand$ and $\opt \ge (1-\varepsilon') \rand$. Then, feeding the output of the above construction to this algorithm implies a $\crp$ algorithm for the decision version of \rc, which by Theorem~\ref{lem:cited_hardness} implies $\np \subseteq \crp$. Since $\rp \subseteq \np$, this implies $\rp \subseteq \crp$, so that $\zpp = \rp \cap \crp = \rp$. Since $\zpp$ is symmetric with respect to ``YES'' and ``NO'' instances, this implies $\zpp = \crp$, so that $\zpp = \np$. 

We now show how to set $\varepsilon$. For the first part of the theorem, we set $\varepsilon = \frac{1}{10} \left( \frac{k}{m} \right)^{1-\delta}$. This can be achieved by choosing $m$ such that $\left(\frac{m}{k}\right)^{\delta} = 20 k z$. Note that this ensures $m = \mbox{poly}(k,z)$ when $\delta$ is a constant, so that the construction runs in polynomial time. For this setting, we have $\varepsilon' \cdot \frac{m+1}{k+1} \le \left( \frac{m+1}{k+1} \right)^{\delta}$, while  $\varepsilon' = \left( \frac{k}{m} \right)^{1-\delta} = \left( \frac{1}{20kz} \right)^{\frac{1-\delta}{\delta}} = o(1)$, completing the proof.

For the second part of the theorem, we set $m = 20 k^{1+\alpha} z$, and $\varepsilon = \frac{1}{10} \frac{1}{k^{\alpha}}$. Again, we have $m = \mbox{poly}(k,z)$, and $\varepsilon' = o(1)$, completing the proof.
\end{proof}

Theorem~\ref{thm:hardness} now implies the following easy corollaries.
\begin{corollary}
Unless $\np = \zpp$, there is no $k^{\alpha}$-approximation to the $1$-Borda score for any constant $\alpha > 0$. Similarly, there is no $\left( \frac{m+1}{k+1} \right)^{1-\delta}$-approximation for any constant $\delta \in (0,1)$.
\end{corollary}

The next corollary adapts the hardness proof to the maximization version of the problem.

\begin{corollary}
For the maximization version of $1$-Borda, there is no polynomial time $\left(1 - \frac{1 - \varepsilon}{k+1} \right)$-approximation for constant $\varepsilon \in (0,1/2)$ unless $\np = \zpp$.
\end{corollary}
\begin{proof}
Set $\varepsilon > 0$ to be a small constant in the proof of Theorem~\ref{thm:main_lb}. Then, in the ``NO'' instance, the maximization score is at most $(m+1) \left(1 - \frac{1 - \varepsilon}{k+1} \right)$, while for the ``YES'' instance, the score is at least $(m+1) \left( 1 - \frac{\varepsilon}{k+1} \right)$. For $\varepsilon \in (0,1/2)$, the approximation factor achievable is therefore at most $\left( 1 - \frac{1-2\varepsilon}{k+1} \right)$, completing the proof.
\end{proof}




\subsection{An Optimal Deterministic Algorithm}
Given the lower bound and the hardness result, an immediate question is whether there is a deterministic rule to achieve the benchmark \rand{}. We answer in the affirmative: The \b{} algorithm~\cite{Banzhaf,dubeyS,Heuristics} can be viewed as a derandomization of \rand{}: Instead of randomly picking a candidate at each iteration, it picks the candidate that gives the best expected performance if the rest of the committee is randomly constructed. It is shown in~\cite{Heuristics} that this algorithm runs in polynomial time. The following theorem implies Theorem~\ref{thm:main_banzhaf}.

\begin{theorem}
$\b \leq \rand$.
\label{thm:banzhaf_better_than_rand}
\end{theorem}
\begin{proof}
Recall that \b{} builds sets $\varnothing = T_0 \subsetneq T_1 \subsetneq \cdots \subsetneq T_k$, where at step $j$, \b{} picks $c_j \in \C \setminus T_{j - 1}$ such that:
\begin{equation}
    \label{eq:b_opt}
c_j = \mbox{argmin}_{c \in \C \setminus T_{j - 1} } \ \sum_{\substack{S \subseteq \C: |S| = k \\ S \supseteq T_{j - 1} \cup \{c\}}} r_{\V}(S).
\end{equation}
We now use induction to show that for any $j \in \{0, 1, \ldots, k\}$,
\[
\frac{1}{\binom{m - j}{k - j}} \sum_{\substack{S \subseteq \C: |S| = k \\ S \supseteq T_j}} r_{\V}(S) \leq \rand,
\]
which is clearly true when $j = 0$, and gives the desired result $\b \leq \rand$ when $j = k$.

For the inductive step, assume it holds for some $j - 1$. Now in the $j^{\text{th}}$ iteration, we have the following inequalities that complete the proof. Here, the first step follows since \b{} picks $c_j$ in step $j$. The second step follows since \b{} solves Eq~(\ref{eq:b_opt}), so that the score from adding $c_j$ beats the average score of adding one of the $m-j+1$ candidates in $C \setminus T_{j-1}$.  The final equality follows since $\binom{m - j + 1}{k - j + 1} = \frac{m-j+1}{k-j+1} \binom{m - j}{k - j}$, and by observing that for any $S \supseteq T_{j-1}$, there are $k-j+1$ choices of $c \in S \setminus T_{j-1}$. 
\begin{align*}
\frac{1}{\binom{m - j}{k - j}} \sum_{\substack{S \subseteq \C: |S| = k \\ S \supseteq T_j}} r_{\V}(S) &= \frac{1}{\binom{m - j}{k - j}} \sum_{\substack{S \subseteq \C: |S| = k \\ S \supseteq T_{j - 1} \cup \{c_j\}}} r_{\V}(S)\\
\leq &\frac{1}{\binom{m - j}{k - j}} \cdot \left(\frac{1}{m-j+1} \sum_{c \in C \setminus T_{j - 1}} \sum_{\substack{S \subseteq \C: |S| = k \\ S \supseteq T_{j - 1} \cup \{c\}}} r_{\V}(S) \right)\\
= &  \frac{1}{\binom{m - j}{k - j}} \cdot \frac{k - j + 1}{m - j + 1} \cdot \left(\frac{1}{k-j+1} \sum_{c \in C \setminus T_{j - 1}} \sum_{\substack{S \subseteq \C: |S| = k \\ S \supseteq T_{j - 1} \cup \{c\}}} r_{\V}(S) \right)\\
= &\frac{1}{\binom{m - j + 1}{k - j + 1}} \sum_{\substack{S \subseteq \C: |S| = k \\ S \supseteq T_{j - 1}}} r_{\V}(S) \leq \rand. \qedhere
\end{align*}
\end{proof}

To complete the proof of Theorem~\ref{thm:main_banzhaf}, for the maximization objective, \b{} achieves a value at least $(m+1) \left(1 -  \frac{1}{k+1} \right)$. Since the maximum possible value is $m$, this implies a $\left(1 - \frac{1}{k+1}\right)$-approximation.

\section{Lower Bound on Committee-Monotone Algorithms for $1$-Borda}
\label{sec:monotone}
Consider the $1$-Borda score. A nice property of \g{} is that it is committee-monotone: In each iteration, the candidate chosen by \g{} only depends on which candidates have been chosen in previous iterations and not on $k$, and thus when $k$ increases, the committee selected by \g{} includes all the candidates \g{} used to select. On the other hand, \b{} does not satisfy committee monotonicity, as the candidates chosen by \b{} does depend on $k$.

This naturally brings up the question: Is there a committee-monotone algorithm which is optimal with respect to the benchmark \rand{}? We answer this question in the negative, by presenting a lower bound of $1.015$ for all committee-monotone algorithms.


\begin{theorem}
For any large enough $m$, there exist instances with $m$ candidates where any committee-monotone algorithm \alg{} satisfies $r_{\V}(T_k) > 1.015 \cdot \rand$ for some value $k \in \{1, 2\}$. Here, $T_k$ is the set of candidates \alg{} chooses when the size of this set is $k$.
\label{thm:monotonicity}
\end{theorem}
\begin{proof}
The construction goes as follows: There are two types of candidates, $X$ and $Y$. Candidates of type $X$ are ranked between $[am, bm]$ by every voter and candidates of type $Y$ are ranked between $[1, am] \cup [bm, m]$ by every voter, where $0 < a < b < 1$ are two parameters. We construct sufficiently many voters so that all candidates of the same type are symmetric (by having all permutations of candidates of type $X$ and those of type $Y$). We want to find proper $a$ and $b$, so that when $k = 1$, the optimal choice is to choose a candidate of type $X$, while when $k = 2$, the optimal choice is to choose two candidates both of type $Y$. This means that no committee-monotone algorithm can produce optimal choice for both the first iteration and the second iteration. We optimize over $a$ and $b$ to find the maximum lower bound.

In particular, the search procedure goes as follows. Let $r_{\V}(Y)$ denote the $1$-Borda score of choosing a candidate of type $Y$; $r_{\V}(XX)$ denote the score of choosing two candidates both of type $X$; and $r_{\V}(XY)$ denote the score of choosing a candidate of type $X$ and a candidate of type $Y$. We can easily see that, when $m$ goes to infinity, up to an $o(1)$ additive error:
\[
\begin{cases}
\frac{1}{m + 1}r_{\V}(Y) = \frac{a}{2} \cdot \Pr[Y \text{ is at } [1, am]] + \frac{1 + b}{2} \cdot \Pr[Y \text{ is at } [bm, m]] = \frac{a}{2} \cdot \frac{a}{1 - (b - a)} + \frac{1 + b}{2} \cdot \frac{1 - b}{1 - (b - a)}\\
\frac{1}{m + 1}r_{\V}(XX) = \frac{2a + b}{3}\\
\frac{1}{m + 1}r_{\V}(XY) = \frac{a}{2} \cdot \Pr[Y \text{ is at } [1, am]] + \frac{a + b}{2} \cdot \Pr[Y \text{ is at } [bm, m]] = \frac{a}{2} \cdot \frac{a}{1 - (b - a)} + \frac{a + b}{2} \cdot \frac{1 - b}{1 - (b - a)}
\end{cases}.
\]

A committee-monotone algorithm either chooses $Y$ in the first iteration, or chooses $XX$ or $XY$ in the first two iterations. Thus, we maximize $\min\left(\frac{2}{m + 1}r_{\V}(Y), \frac{3}{m + 1}r_{\V}(XX), \frac{3}{m + 1}r_{\V}(XY)\right)$ (note that the value on the numerator corresponds to the value of $k + 1$) over $0 < a < b < 1$, and find that, for $a = 0.377$ and $b = 0.552$, it achieves a lower bound greater than $1.015$.
\end{proof}

\section{Connection to the Core}
\label{sec:core}
We now consider the relationship between the core and $1$-Borda score. In particular, we show that the core achieves a stronger notion of proportionality than $1$-Borda: any $\alpha$-approximate core solution has $1$-Borda score at most $\alpha \cdot \frac{k + 1}{k} \cdot \rand$, while neither the optimal solution $\opt$ nor the algorithms \g{} and \b{} lies in an $o(k)$-approximate core.

Recall that we say a committee $T$ of size $k$ is in the $\alpha$-approximate core if there is no blocking candidate strictly preferred by at least $\alpha \cdot \frac{n}{k}$ voters over $T$. See Eq~(\ref{eq:core}) for a formal definition. In this section, we investigate the relationship between $1$-Borda and the core.

First, we present in the following theorem the implication of the core for $1$-Borda score.

\begin{theorem}
Any committee $T$ in the $\alpha$-approximate core satisfies $r_{\V}(T) \leq \alpha \cdot \frac{k + 1}{k} \cdot \rand$.
\label{thm:core_to_borda}
\end{theorem}
\begin{proof}
As $T$ is in the $\alpha$-approximate core, there is no deviation of size $\frac{\alpha n}{k}$, i.e., there is no candidate ranked above all candidates in $T$ for $\frac{\alpha n}{k}$ voters. Therefore,
\[
\frac{1}{m - k} \sum_{v \in \V} (r_v(T) - 1) \leq \frac{\alpha n}{k}
\]
by a counting argument. Thus,
\[
r_{\V}(T) = \frac{1}{n} \sum_{v \in \V} r_v(T) \leq (m - k) \cdot \frac{\alpha}{k} = \alpha \cdot \frac{k + 1}{k} \cdot \frac{m + 1}{k + 1}. \qedhere
\]
\end{proof}

Naturally we ask: Does the reverse statement -- a good approximation to $\rand$ for the $1$-Borda score gives a good approximation to the core -- hold as well? It turns out that the answer is no.

\begin{example}
Let $n = 3 \cdot (m - 2)!$, where $m$ is sufficiently large. $c_1$ and $c_2$ are two ``critical'' candidates, and the remaining $m - 2$ are ``dummy'' candidates. For the first $\frac{n}{3}$ voters, $c_1$ is their top choices and $c_2$ is their second choices. For the second $\frac{n}{3}$ voters, $c_1$ is their bottom choices and $c_2$ is their top choices. For the last $\frac{n}{3}$ voters, $c_1$ is their bottom choice and $c_2$ is their second bottom choice. We fill the rest of their preferences with all permutations of the dummy candidates. This example is illustrated in Figure~\ref{fig:not_core}.
\label{ex:not_core}
\end{example}

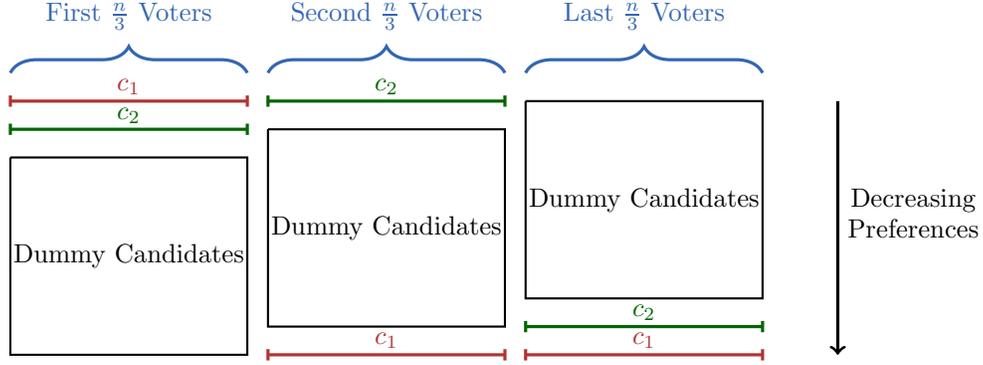
\begin{figure}
\centering
\begin{tikzpicture}[yscale = 3.75, xscale = 5]

\draw [decorate, very thick, color3, decoration={brace,amplitude=10pt}] (-1, 2) -- (-0.37, 2);
\draw [decorate, very thick, color3, decoration={brace,amplitude=10pt}] (-0.315, 2) -- (0.315, 2);
\draw [decorate, very thick, color3, decoration={brace,amplitude=10pt}] (0.37, 2) -- (1, 2);

\draw [very thick, color1] (-1, 1.9) -- (-0.37, 1.9);
\draw [very thick, color1] (-1, 1.88) -- (-1, 1.92);
\draw [very thick, color1] (-0.37, 1.88) -- (-0.37, 1.92);

\draw [very thick, color2] (-1, 1.8) -- (-0.37, 1.8);
\draw [very thick, color2] (-1, 1.78) -- (-1, 1.82);
\draw [very thick, color2] (-0.37, 1.78) -- (-0.37, 1.82);

\draw [thick] (-1, 1.7) -- (-0.37, 1.7) -- (-0.37, 1.0) -- (-1, 1.0) -- (-1, 1.7);

\draw [very thick, color1] (-0.315, 1.0) -- (0.315, 1.0);
\draw [very thick, color1] (-0.315, 0.98) -- (-0.315, 1.02);
\draw [very thick, color1] (0.315, 0.98) -- (0.315, 1.02);

\draw [very thick, color2] (-0.315, 1.9) -- (0.315, 1.9);
\draw [very thick, color2] (-0.315, 1.88) -- (-0.315, 1.92);
\draw [very thick, color2] (0.315, 1.88) -- (0.315, 1.92);

\draw [thick] (-0.315, 1.8) -- (0.315, 1.8) -- (0.315, 1.1) -- (-0.315, 1.1) -- (-0.315, 1.8);

\draw [very thick, color1] (0.37, 1.0) -- (1, 1.0);
\draw [very thick, color1] (0.37, 0.98) -- (0.37, 1.02);
\draw [very thick, color1] (1, 0.98) -- (1, 1.02);

\draw [very thick, color2] (0.37, 1.1) -- (1, 1.1);
\draw [very thick, color2] (0.37, 1.08) -- (0.37, 1.12);
\draw [very thick, color2] (1, 1.08) -- (1, 1.12);

\draw [thick] (0.37, 1.9) -- (1, 1.9) -- (1, 1.2) -- (0.37, 1.2) -- (0.37, 1.9);

\draw [very thick, ->] (1.2, 1.9) -- (1.2, 1);

\node [color3] at (-0.685, 2.2) {First $\frac{n}{3}$ Voters};
\node [color3] at (0, 2.2) {Second $\frac{n}{3}$ Voters};
\node [color3] at (0.685, 2.2) {Last $\frac{n}{3}$ Voters};

\node [color1] at (-0.685, 1.95) {$c_1$};
\node [color2] at (-0.685, 1.85) {$c_2$};
\node at (-0.685, 1.35) {Dummy Candidates};

\node [color1] at (0, 1.05) {$c_1$};
\node [color2] at (0, 1.95) {$c_2$};
\node at (0, 1.45) {Dummy Candidates};

\node [color1] at (0.685, 1.05) {$c_1$};
\node [color2] at (0.685, 1.15) {$c_2$};
\node at (0.685, 1.55) {Dummy Candidates};

\node at (1.40, 1.45) {Preferences};
\node at (1.40, 1.55) {Decreasing};

\end{tikzpicture}
\caption{Illustration of Voters' Preferences in Example~\ref{ex:not_core}}
\label{fig:not_core}
\end{figure}

\begin{theorem}
The solutions of \opt{}, \g{} and \b{} do not lie in an $o(k)$-approximate core in Example~\ref{ex:not_core}. 
\label{thm:borda_to_core}
\end{theorem}
\begin{proof}
Let $k = \sqrt{m} - 1$ in Example~\ref{ex:not_core}. We show all of \opt{}, \g{} and \b{} chooses $c_2$ and $k - 1$ dummy candidates. In this solution, $\frac{n}{3}$ voters could deviate to $c_1$, showing that it does not lie in a $\frac{k}{3}$-approximate core.

\paragraph{Proof for \opt{}}  
We compare the resulting $s$-Borda score for all possible schemes: choosing $c_1$, $c_2$, and $k - 2$ dummy candidates; choosing $c_1$ and $k - 1$ dummy candidates; choosing $c_2$ and $k - 1$ dummy candidates; and choosing $k$ dummy candidates. Let $D_j$ be a set consisting of $j$ dummy candidates. Then, we have:
\begin{align*}
r_{\V}(D_{k - 2} \cup \{c_1\} \cup \{c_2\}) = \frac{m + 1}{3(k - 1)} + \frac{2}{3}, &\quad r_{\V}(D_{k - 1} \cup \{c_1\}) = \frac{2(m + 1)}{3k} + \frac{1}{3},\\
r_{\V}(D_{k - 1} \cup \{c_2\}) = \frac{m + 1}{3k} + 1, &\quad r_{\V}(D_k) = \frac{m + 1}{k + 1}.
\end{align*}

For $k = \sqrt{m} - 1$, we have:
$$r_{\V}(D_{k - 1} \cup \{c_2\}) < r_{\V}(D_{k - 2} \cup \{c_1\}\cup\{c_2\}) <  r_{\V}(D_{k - 1} \cup \{c_1\}) < r_{\V}(D_k).$$

Thus, \opt{} chooses $c_2$ and $k - 1$ dummy candidates without choosing $c_1$.

\paragraph{Proof for \g{}}
For the first iteration, \g{} chooses $c_2$. We will show that, for the next $\sqrt{m} - 2$ iterations, \g{} chooses the dummy candidates and does not choose $c_1$. Suppose we have chosen $j - 1$ candidates, where $j - 1 \leq \sqrt{m} - 1$, and the current set of candidates is $T_{j - 1}$. Then, we have:
\[
r_{\V}(T_{j - 1}) - r_{\V}(T_{{j - 1}} \cup \{c_1\}) = \frac{1}{3},
\]
\[
r_{\V}(T_{j - 1}) - r_{\V}(T_{j - 1} \cup \{c_j\}) = \frac{1}{3}\left(\frac{m + 1}{j} - \frac{m + 1}{j - 1}\right) = \frac{m + 1}{3j(j - 1)} > \frac{1}{3}, \forall c_j \in \C \setminus T_{j - 1}, c_j \neq c_1.
\]
which shows that for the $\sqrt{m} - 2$ iterations after the first iteration, \g{} chooses dummy candidates.

\paragraph{Proof for \b{}}
Let $T_j$ be the set of candidates produced by \b{} after $j$ iterations. Recall that by \b{}, in the $j^{\text{th}}$ iteration, we pick $c_j \in \C \setminus T_{j - 1}$ that minimizes:
\[
\sum_{\substack{S \subseteq \C: |S| = k \\ S \supseteq T_{j - 1} \cup \{c_j\}}} r_{\V}(S).
\]

Clearly, \b{} chooses $c_2$ in the first iteration, because, as we have shown in the proof for \opt{}, for $k = \sqrt{m} - 1$, choosing $c_2$ always yields better result than not choosing $c_2$. 

Then, we show that \b{} chooses dummy candidates for the next $\sqrt{m} - 2$ iterations. Assume at $(j - 1)^{\text{th}}$ iteration, we have chosen $j - 2$ dummy candidates and $c_2$. As we have shown in the proof for \opt{}, for $k = \sqrt{m} - 1$, we have $r_{\V}(T_{k - 1} \cup \{c_j\}) < r_{\V}(T_{k - 1} \cup \{c_1\})$, $\forall c_j \in \C \setminus T_{k - 1}, c_j \neq c_1$, where $T_{k - 1}$ is a set consisting of $c_2$ and $k - 2$ dummy candidates. This implies that the candidate that minimizes the above objective is dummy candidate but not $c_1$. Thus, for the $j^{\text{th}}$ iteration, \b{} also chooses a dummy candidate, and by inductive principle, \b{} chooses $c_2$ and $k - 1$ dummy candidates in $k$ iterations.
\end{proof}

Theorem~\ref{thm:core_to_borda} and Theorem~\ref{thm:borda_to_core} together establishes that the core achieves a stronger notion of proportionality than $1$-Borda.

\section{The $s$-Borda Score}
In this section, we extend our analysis of the greedy algorithms to $s$-Borda score, and show how to significantly improve on the \g{} and \b{} rules via linear programming. 

Recall that $\rand = \frac{m+1}{k+1}$ and choosing a random committee of size $k$ yields expected score $\frac{s(s + 1)}{2} \cdot \rand$. As a derandomization, \b{} has score at most this value similar to Theorem~\ref{thm:banzhaf_better_than_rand}. Let $\opt$ denote the best possible $s$-Borda score. Considering the instance with one voter for each permutation of candidates as its preference ordering, we have the following proposition:

\begin{proposition}
For any $s$, $m$ and $k$, there exists instances where $\opt =\frac{s(s + 1)}{2} \cdot \rand$.
\end{proposition}

We first consider a natural extension of \g{} in the $1$-Borda case. In Appendix~\ref{app:sborda}, we show that it achieves an $s$-Borda score at most $2s^2\cdot \rand$ (Theorem~\ref{thm:s_borda}), which is within a factor of $\frac{4s}{s + 1}$ of the \b{} rule. We then show that this bound cannot be improved even when $\opt$ is small. However, unlike the $1$-Borda case, there is no fundamental barrier to an improved algorithm when $\opt$ is small, and we present such an algorithm in Section~\ref{sec:lp}.

\subsection{The \g{} Algorithm} 
\label{sec:sborda_greedy}
\label{sec:sborda_g_ub}
\label{sec:sborda_g_lb}
The \g{} algorithm follows exactly the same procedure as for $1$-Borda, except that we now compute the score based on $s$-Borda. We present an upper bound for \g{} in the following theorem. Since the proof is very similar to the $s=1$ case, we present it in Appendix~\ref{app:sborda}.

\begin{theorem}[Proved in Appendix~\ref{app:sborda}]
\label{thm:s_borda}
$\g{} \leq 2s^2\cdot \rand$.
\end{theorem}

\paragraph{Lower Bound for Small $\opt$.} 
In general, $\opt = \Omega(s^2) \cdot \rand$, in which case the analysis of greedy is tight to within a constant factor. The question we now ask is: Does \g{} always perform better when $\opt$ is small? We answer this in the negative.

\begin{theorem}
\label{thm:sbordalb}
There exists an instance where $\opt = O(s^2) = o(1) \cdot \rand$, while the score of \g{} is $\Omega(s^2)\cdot \rand$.
\end{theorem}
To prove this lower bound, we use the following instance.
\begin{example}
Let $m = \omega(k)$, $k = \omega(s)$, and $n = \frac{k}{s}(m - k)!$. There are $k$ ``critical'' candidates $c_1, c_2, \ldots, c_k$, while the remaining $m - k$ are ``dummy'' candidates. Candidate $c_{i(k/s) + j}$ is the $(i + 1)^{\text{st}}$ choice of the $j^{\text{th}}$ $\frac{k}{s}$ voters, $\forall i \in \{0, 1, \ldots, s - 1\}, j \in \{1, 2, \ldots, \frac{k}{s}\}$. Aside from the first $s$ rows, the critical candidates lie at the very bottom. For each group of $\frac{k}{s}$ voters, we fill the rest of the preferences with all permutations of the dummy candidates. This example is illustrated in Figure~\ref{figure:greedy_bad}. 
\label{example:greedy_bad}
\end{example}

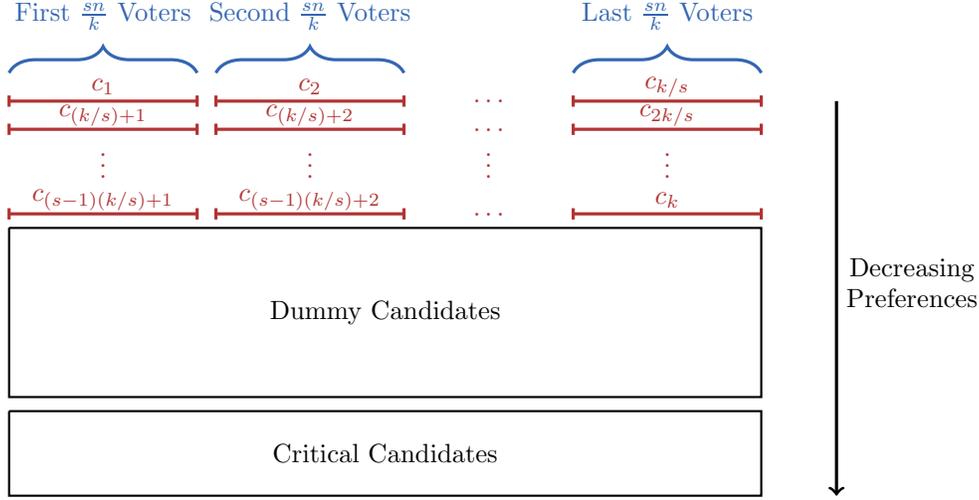
\begin{figure}
\centering
\begin{tikzpicture}[yscale = 3.75, xscale = 5]

\draw [decorate, very thick, color3, decoration={brace,amplitude=10pt}] (-1, 2) -- (-0.5, 2);
\node [color3] at (-0.75, 2.2) {First $\frac{sn}{k}$ Voters};
\draw [decorate, very thick, color3, decoration={brace,amplitude=10pt}] (-0.45, 2) -- (0.05, 2);
\node [color3] at (-0.2, 2.2) {Second $\frac{sn}{k}$ Voters};
\draw [decorate, very thick, color3, decoration={brace,amplitude=10pt}] (0.5, 2) -- (1, 2);
\node [color3] at (0.75, 2.2) {Last $\frac{sn}{k}$ Voters};

\draw [very thick, color1] (-1, 1.9) -- (-0.5, 1.9);
\draw [very thick, color1] (-1, 1.88) -- (-1, 1.92);
\draw [very thick, color1] (-0.5, 1.88) -- (-0.5, 1.92);
\node [color1] at (-0.75, 1.95) {$c_1$};

\draw [very thick, color1] (-0.45, 1.9) -- (0.05, 1.9);
\draw [very thick, color1] (-0.45, 1.88) -- (-0.45, 1.92);
\draw [very thick, color1] (0.05, 1.88) -- (0.05, 1.92);
\node [color1] at (-0.2, 1.95) {$c_2$};

\draw [very thick, color1] (0.5, 1.9) -- (1, 1.9);
\draw [very thick, color1] (0.5, 1.88) -- (0.5, 1.92);
\draw [very thick, color1] (1, 1.88) -- (1, 1.92);
\node [color1] at (0.75, 1.95) {$c_{k/s}$};

\node [color1] at (0.275, 1.9) {$\ldots$};

\draw [very thick, color1] (-1, 1.8) -- (-0.5, 1.8);
\draw [very thick, color1] (-1, 1.78) -- (-1, 1.82);
\draw [very thick, color1] (-0.5, 1.78) -- (-0.5, 1.82);
\node [color1] at (-0.75, 1.85) {$c_{(k/s)+1}$};

\draw [very thick, color1] (-0.45, 1.8) -- (0.05, 1.8);
\draw [very thick, color1] (-0.45, 1.78) -- (-0.45, 1.82);
\draw [very thick, color1] (0.05, 1.78) -- (0.05, 1.82);
\node [color1] at (-0.2, 1.85) {$c_{(k/s) + 2}$};

\draw [very thick, color1] (0.5, 1.8) -- (1, 1.8);
\draw [very thick, color1] (0.5, 1.78) -- (0.5, 1.82);
\draw [very thick, color1] (1, 1.78) -- (1, 1.82);
\node [color1] at (0.75, 1.85) {$c_{2k/s}$};

\node [color1] at (0.275, 1.8) {$\ldots$};

\node [color1] at (-0.75, 1.7){$\vdots$};

\node [color1] at (-0.2, 1.7){$\vdots$};

\node [color1] at (0.75, 1.7){$\vdots$};

\node [color1] at (0.275, 1.7){$\vdots$};

\draw [very thick, color1] (-1, 1.5) -- (-0.5, 1.5);
\draw [very thick, color1] (-1, 1.48) -- (-1, 1.52);
\draw [very thick, color1] (-0.5, 1.48) -- (-0.5, 1.52);
\node [color1] at (-0.75, 1.55) {$c_{(s-1)(k/s)+1}$};

\draw [very thick, color1] (-0.45, 1.5) -- (0.05, 1.5);
\draw [very thick, color1] (-0.45, 1.48) -- (-0.45, 1.52);
\draw [very thick, color1] (0.05, 1.48) -- (0.05, 1.52);
\node [color1] at (-0.2, 1.55) {$c_{(s-1)(k/s) + 2}$};

\draw [very thick, color1] (0.5, 1.5) -- (1, 1.5);
\draw [very thick, color1] (0.5, 1.48) -- (0.5, 1.52);
\draw [very thick, color1] (1, 1.48) -- (1, 1.52);
\node [color1] at (0.75, 1.55) {$c_{k}$};

\node [color1] at (0.275, 1.5) {$\ldots$};

\draw [thick] (-1, 1.45) -- (1, 1.45) -- (1, 0.85) -- (-1, 0.85) -- (-1, 1.45);
\node at (0, 1.15) {Dummy Candidates};

\draw [thick] (-1, 0.8) -- (1, 0.8) -- (1, 0.5) -- (-1, 0.5) -- (-1, 0.8);
\node at (0, 0.65) {Critical Candidates};

\draw [very thick, ->] (1.2, 1.9) -- (1.2, 0.5);

\node at (1.40, 1.30) {Decreasing};
\node at (1.40, 1.20) {Preferences};

\end{tikzpicture}
\caption{Illustration of Voters' Preferences in Example~\ref{example:greedy_bad}}
\label{figure:greedy_bad}
\end{figure}

In this instance, \opt{} is clearly $O(s^2)$ by choosing all the critical candidates. We now show that \g{} achieves its worst-case bound even on this instance.

\begin{proposition}
In Example~\ref{example:greedy_bad}, $\g{} = \Omega(s^2)\cdot \rand$.
\end{proposition}
\begin{proof}
For the first $s$ iterations, \g{} chooses dummy candidates: as $m \gg k$, choosing a critical candidate adds $\frac{s - 1}{s}(m + 1)$ to the score, while choosing a dummy candidate adds only $\frac{1}{2}(m + 1)$.

Then, we show that, for the first $\frac{k}{2}$ iterations, \g{} chooses dummy candidates. Assume at $(j - 1)^{\text{th}}$ iteration, where $s \leq j - 1 < \frac{k}{2}$, we have chosen $j - 1$ dummy candidates, and the set of candidates is $T_{j - 1}$. Then, we have:
\[
r_{\V}(T_{j - 1}) - r_{\V}(T_{j - 1} \cup \{c_\mathrm{critical}\}) \leq \frac{s}{k}\cdot \frac{s}{j}(m + 1) = \frac{s^2}{kj}(m + 1),
\]
and
\[
r_{\V}(T_{j - 1}) - r_{\V}(T_{j - 1} \cup \{c_\mathrm{dummy}\}) = \frac{s(s + 1)}{2j}(m + 1) - \frac{s(s + 1)}{2(j + 1)}(m + 1) = \frac{s(s + 1)}{2j(j + 1)}(m + 1),
\]
where $c_\mathrm{critical}$ is some critical candidate and $c_\mathrm{dummy}$ is some dummy candidate. This is because if we choose a critical candidate, then for $\frac{s}{k}$ fraction of the voters, the bottom-ranked dummy candidate will be dropped, while the critical candidate will be added. Since we have chosen $j - 1$ dummy candidates, the bottom-ranked dummy candidate has average rank $\frac{s}{j}(m + 1)$. In other words, for $\frac{s}{k}$ fraction of the voters, we drop a candidate at rank $\frac{s}{j}(m + 1)$ and gain a candidate at the top, while for the other voters, the top $s$ candidates remain unchanged. If we choose a dummy candidate instead, the average score goes from $\frac{s(s + 1)}{2j}(m + 1)$ to $\frac{s(s + 1)}{2(j + 1)}(m + 1)$.

For $j \leq \frac{k}{2}$, we have:
$$\frac{s(s + 1)}{2j(j + 1)}(m + 1) > \frac{s^2}{kj}(m + 1),$$
and thus \g{} chooses a dummy candidate in the $j^{\text{th}}$ iteration as well. Thus, by inductive principle, \g{} chooses dummy candidates for at least $\frac{k}{2}$ iterations.

However, this implies that we can choose at most $\frac{k}{2}$ critical candidates. Suppose for the $i^{\text{th}}$ $\frac{s}{k}$ voters, there are $x_i$ critical candidates among the top $s$ candidates. We have:
$$\sum_{i = 1}^{k/s}x_i \leq \frac{k}{2}.$$

Let $T_k$ denote the final set of candidates. As we choose at most $\frac{k}{2}$ critical candidates, at least $\frac{k}{2}$ candidates must be chosen, and we derive a lower bound for $r_\V(T_k)$ based on this. We have:
\begin{align*}
r_{\V}(T_k) &\geq \frac{s}{k}\left(\sum_{i = 1}^{k/s}\sum_{j = 1}^{s - x_i}j\right)\cdot \rand \geq \frac{s}{2k}\left(\sum_{i = 1}^{k/s}(s - x_i)^2\right) \cdot \rand \\
& \geq \frac{s}{2k} \frac{\left(\sum_{i = 1}^{k/s}(s - x_i)\right)^2}{k/s}\cdot \rand \geq \frac{s^2}{8}\cdot \rand.
\end{align*}

Recall that given $n$ voters whose preference structures include all permutations of the $m$ candidates, when we choose $k$ candidates out of them, the average contribution of the $i^{\text{th}}$-ranked candidates for each voter to $r_\V(T_k)$ is $i\cdot \rand$. The first inequality is by applying the above fact on each set of $\frac{s}{k}$ voters whose preference structures include all permutations. The third inequality is by Cauchy-Schwarz inequality. The last inequality is because $\sum_{i = 1}^{k/s}x_i \leq \frac{k}{2}$. Thus, we can conclude that $r_{\V}(T_k) = \Omega(s^2)\cdot \rand$.
\end{proof}

This shows that \g{} can perform as bad as random even when \opt{} is small and thus motivates the improved guarantee in Section~\ref{sec:lp}.

\subsection{An Improved Algorithm via LP Rounding}
\label{sec:lp}
As mentioned above, \g{} can hit its worst-case bound of $\Omega(s^2)\cdot \rand$ even when \opt{} is actually small. We know that for the case of $1$-Borda, no polynomial-time algorithm can do better. Now the question is, can a different algorithm do better in the case of $s$-Borda for $s = \omega(1)$? We answer this question in the affirmative by presenting an algorithm that is based on dependent rounding of an LP relaxation combined with uniform random sampling, which provides nontrivial improvement when $\opt$ is small. In particular, it achieves expected score at most $3\cdot \opt + O(s^{3/2}\log s)\cdot \rand$. 

\subsubsection{LP-Rounding-Based Algorithm}
The following linear program (based on~\cite{cornuejols1983,LinV,CharikarGTS,fault,Byrka}) is a natural relaxation for the $s$-Borda problem.
\begin{mini*}
{}{\sum_{i = 1}^n\sum_{\ell = 1}^s\sum_{j=1}^{m} x_{ij}^{\ell} \cdot r_{v_i}(c_j),}{}{}
\addConstraint{\sum_{j=1}^m\:} {y_j}{=k}
\addConstraint{\sum_{\ell = 1}^k\:}{x_{ij}^{\ell}}{\leq y_j,}{\mkern46mu i \in \{1, \ldots, n\}, \:\:j\in\{1 ,\ldots, m\}} 
\addConstraint{\sum_{j = 1}^m\:}{x_{ij}^{\ell}}{\geq 1,}{\mkern53mu i \in \{1, \ldots, n\},\:\: \ell\in \{1, \ldots, k\}}
\addConstraint{y_j, \:}{x_{ij}^{\ell}}{\in [0, 1],}{\mkern26mu i \in \{1, \ldots, n\}, \:\: j \in \{1, \ldots, m\}, \:\: \ell \in \{1, \ldots, k\}.}
\end{mini*}

Variable $y_j$ denotes how much candidate $c_j$ is chosen; integral values $1$ and $0$ mean choosing and not choosing candidate $c_j$, respectively. The first constraint encodes choosing exactly $k$ candidates. We copy each voter $k$ times, and the $\ell^{\text{th}}$ copy of the voter $v_i$ is assigned to the $\ell^{\text{th}}$-preferred chosen candidate. Variable $x_{ij}^{\ell}$ denotes how much the $\ell^{\text{th}}$ copy of voter $v_i$ is assigned to candidate $c_j$. The second constraint prevents a voter from being assigned to a candidate that is not chosen. The third constraint ensures that each copy of the voter is assigned to some candidate. The objective function computes the $s$-Borda score.

We will use dependent rounding~\cite{dependent} to round this LP solution. There is a catch though: Dependent rounding can cause a deficit in around $\tilde{O}(\sqrt{s})$ candidates from the top $s$ that are fractionally chosen by the LP. Since any solution must account for the top $s$ scores, we need to ensure these ``deficit'' candidates do not increase the score too much. Towards this end, we scale down the LP solution, and choose enough candidates uniformly at random so that these candidates can absorb the deficit. However, such scaling creates a further deficit that will have to be absorbed by random sampling. We find that the right trade-off is achieved by scaling down by a factor of $(1-\frac{1}{\sqrt{s}})$.

Without further ado, the overall algorithm works as follows:
\begin{enumerate}
    \item Solve the above linear program and let $\tilde{y}$ denote the optimal solution.
    \item For $j = 1,2,\ldots, m$, let $y_j = (1 - \frac{1}{\sqrt{s}}) \tilde{y}_j$. Note that $\sum_{j=1}^m y_j = k (1 - \frac{1}{\sqrt{s}})$.
    \item Apply dependent rounding~\cite{dependent} on the variables $\{y_j\}$ so that exactly $k (1-\frac{1}{\sqrt{s}})$ candidates are chosen. Let $T_1$ denote the set of chosen candidates.
    \item Finally choose a set $T_2$ of $\frac{k}{\sqrt{s}}$ candidates uniformly at random from $\C \setminus S$ and output $T = T_1 \cup T_2$. 
\end{enumerate}

We will show the following theorem:

\begin{theorem}
\label{thm:lp}
When $m = \omega(k)$, $k = \omega(s^{3/2})$, and $s = \omega(1)$, we have:
$$\E[r_{\V}(T)] \leq 3\opt +O(s^{3/2}\log s)\cdot \rand.$$
\end{theorem}

\subsubsection{Analysis: Proof of Theorem~\ref{thm:lp}}
First consider dependent rounding on $\{y_j\}$. Let $Y_j$ denote the random variable which returns $1$ if $y_j$ is rounded to $1$ and $0$ if $y_j$ is rounded to $0$. Note that $\E[Y_j] = y_j$ for all candidates $j \in \C$. The following lemma is an easy consequence of Chernoff bounds:

\begin{lemma}
\label{lem:bound}
For any subset of candidates $\{c_{j_1}, \ldots, c_{j_\ell}\}$, let $W = \sum_{t = 1}^\ell y_{j_t}$. If $W = \omega(1)$, then 
$$\Pr\left[\sum_{t = 1}^\ell Y_{j_t} \in \left(W - 9\sqrt{W\log W}, W + 9\sqrt{W\log W}\right)\right] \geq 1 - \frac{2}{W^3}.$$
\end{lemma}

 We now compute the expected $s$-Borda score for each voter. Towards this end, we partition the candidates into buckets with geometrically decreasing sum of $y_i$ values, and account for the expected score generated by dependent rounding in each bucket against the LP value of the subsequent bucket. Lemma~\ref{lem:bound} will ensure the number of candidates chosen from each bucket is close to the LP value, and the deficit gets taken care of by the uniformly randomly chosen candidates.

For simplicity of notation, let $\eta = \log_2 \frac{\sqrt{s}}{2}$. Fix a voter $v_i$, and suppose its preference order is $c_{i_1} \succ c_{i_2} \succ \ldots \succ c_{i_m}$. Recall that $\{\tilde{y}, x\}$ is the LP solution. The values $x_{ij}$ in the LP are set as follows: Consider the prefix of the ordering such that $\sum_{t=1}^{\ell} \tilde{y}_{i_t} \le s$ and $\sum_{t=1}^{\ell+1} \tilde{y}_{i_t} > s$. The LP sets $x_{i i_t} = \tilde{y}_{i_t}$ for $t \le \ell$, and sets $x_{i i_{\ell+1}} = s - \sum_{t=1}^{\ell} \tilde{y}_{i_t}$. The contribution of $v_i$ to the LP objective is therefore
\begin{equation}
    \label{eq:opt1}
\opt_i =  \sum_{t=1}^{\ell} r_{v_i}(c_{i_t}) \tilde{y}_{i_t} + r_{v_i}(c_{i_{\ell+1}}) \left(s - \sum_{t=1}^{\ell} \tilde{y}_{i_t} \right) \ge \sum_{t=1}^{\ell} r_{v_i}(c_{i_t}) \tilde{y}_{i_t}.
\end{equation}

Consider the first $\ell$ candidates in the above ordering. We have $\sum_{j=1}^{\ell} \tilde{y}_j \ge s - 1$, so that
\begin{equation}
    \label{eq:sumy}
\sum_{j=1}^{\ell} y_j \ge (s-1) \left(1 - \frac{1}{\sqrt{s}} \right) \ge s - (\sqrt{s}+1) \ge s - 2 \sqrt{s}.
\end{equation}
We split these $\ell$ candidates into sets $\C_1, \ldots, \C_{\eta}$ as follows: We walk down the preference order of $v_i$. We take $\C_1$ as the set of candidates whose $y$-values sum to $\frac{s}{2}$;  $\C_2$ as the next set of candidates whose  $y$-values sum to $\frac{s}{4}$, and so on until $\C_{\eta}$, whose sum of $y$-values is $\frac{s}{2^\eta} = 2\sqrt{s}$.  Now the sum of $y$-values of all candidates in $\{\C_1, \ldots, \C_{\eta}\}$ is exactly $s - 2\sqrt{s}$. Formally, we define 
\[
\theta_0 = 0, \qquad \theta_h = \min\left\{q \ \bigg| \ \sum_{t = 1}^q y_{i_t} \geq \left(1 - \frac{1}{2^h}\right)s \right\}, \:\:\forall h \in \{1, \ldots, \eta\},
\]
and correspondingly define the sets $\{\C_1, \ldots, \C_{\eta}\}$ as:
$$\C_h = \{c_{i_{q}} \mid \theta_{h - 1} < q \leq \theta_h\}, \:\: \forall h \in \:\{1, \ldots, \eta\}.$$

For all $h \in \{1, \ldots, \eta\}$, let $y_{\C_h} = \sum_{c_j \in \C_h}y_j$ and  $Y_{\C_h} = \sum_{c_j \in \C_h}Y_j$. Note that $y_{\C_h}$ decreases by a factor of $2$ as $h$ increases. Now consider the outcome of the dependent rounding procedure for each of the sets $\C_1, \ldots \C_{\eta}$. We say the rounding fails for $v_i$ if there exists $h \in \{1, \ldots, \eta-1\}$ such that the number $Y_{\C_h}$ of chosen candidates  in $\C_h$ is not in range $y_{\C_h} \pm 9 \sqrt{y_{\C_h}\log y_{\C_h}}$. We will not consider $\C_{\eta}$ when defining failure, and will deal with this set separately. 

Let $\F$ denote the failure event. We now bound the probability of the event $\F$ for voter $v_i$.
\begin{lemma}
$\Pr[\F] \leq s^{-3/2}.$
\end{lemma}
\begin{proof}
By union bound applied to Lemma~\ref{lem:bound}, we have:
\begin{align*}
\Pr[\F] & \leq \sum_{h = 1}^{\eta-1}\frac{2}{y_{\C_h}^3} \leq \frac{3}{y_{\C_{\eta-1}}^3} \le s^{-3/2},
\end{align*}
where we have used that $\{y_{\C_h}\}$ is a geometrically decreasing sequence, and that $y_{\C_{\eta-1}} = 4 \sqrt{s}$. 
\end{proof}

We are now ready to compute the expected score for $v_i$ in our algorithm. Recall that $T$ denotes the set of chosen candidates and $\opt_i$ denotes the $s$-Borda score for $v_i$ in the LP solution. Let $\mathsc{Bad}$ denote the expected $s$-Borda score for $v_i$ in the event $\F$, and $\mathsc{Good}$ denote the expected score otherwise. We will bound these separately below.

\begin{lemma}
\label{lem:bad}
$\mathsc{Bad} \le O(s^{5/2})\cdot \rand$.
\end{lemma}
\begin{proof}
If $\F$ happens, the final solution is still at least as good as choosing the $\frac{k}{\sqrt{s}}$ random candidates in Step (4) of the algorithm. Note that since we assumed $k = \omega(s^{3/2})$, we have $\frac{k}{\sqrt{s}} \ge s$, so that we will have chosen enough random candidates to fill up at least $s$ positions for computing $s$-Borda score. Further, since we assume that $m =  \omega(k)$, the score of the solution will at most double had we assumed these $\frac{k}{\sqrt{s}}$ candidates are chosen randomly from the entire set of $m$ candidates instead of from the remaining $m - k(1 - \frac{1}{\sqrt{s}})$ candidates after dependent rounding. Thus, we have:
$$\mathsc{Bad} \leq 2 \E_{T \subseteq C, |T| = k/\sqrt{s}}[r_{v_i}(T)] \le 2 \frac{s(s+1)}{2} \frac{m+1}{\frac{k}{\sqrt{s}} + 1}\leq 4s^{3/2}(s + 1)\cdot \rand,$$
which yields that $\mathsc{Bad} \leq O(s^{5/2})\cdot \rand$.
\end{proof}

\begin{lemma}
\label{lem:good}
$\mathsc{Good} \leq 3\opt_i +O(s^{3/2}\log s)\cdot \rand$.
\end{lemma}
\begin{proof}
Suppose $\F$ does not happen. Denote the set of candidates chosen by the algorithm from $\{\C_1, \ldots, \C_{\eta-1}\}$ as $T_1$, and the randomly chosen $\frac{k}{\sqrt{s}}$ candidates as $T_2$. Therefore $T = T_1 \cup T_2$. From Eq~(\ref{eq:sumy}), we have 
$$\sum_{h=1}^{\eta-1} y_{\C_h} = s - 2\sqrt{s} - y_{\C_{\eta}} = s - 4\sqrt{s}.$$ 
Since $\F$ does not happen, we have:
$$|T_1| \geq s - 4 \sqrt{s} - \sum_{h = 1}^{\eta-1} \sqrt{y_{\C_h}\log y_{\C_h}} \geq s - 4 \sqrt{s} - \sum_{h = 1}^{\eta-1} \sqrt{\frac{s}{2^h}\log s} \geq s - 4 \sqrt{s} - (\sqrt{2} + 1)\sqrt{s\log s}.$$

Denote $u = 4 \sqrt{s} + (\sqrt{2} + 1)\sqrt{s\log s} = O(\sqrt{s\log s})$, so that $|T_1| \ge s - u$.  The quantity $u$ is the total ``deficit'' in candidates from the top $s$ that is caused by scaling the LP and dependent rounding. We make up this deficit using the set $T_2$. Specifically, consider the subsets, 
$$T_1^* = \argmin_{Q \subseteq T_1, |Q| = s - u}\sum_{c \in Q}r_{v_i}(c) \qquad \mbox{and} \qquad T_2^* = \argmin_{Q \subseteq T_2, |Q| = u}\sum_{c \in Q}r_{v_i}(c).$$ 
Note that $T_1^* \subseteq T_1$, and $T_2^* \subseteq T_2$. We will evaluate the score of these subsets of candidates, which will be an upper bound on the score of the algorithm. Towards this end, we define $\mu_h$ as the scaled LP score of $\C_{h}$, that is:
$$\mu_h = \frac{\sum_{j = \theta_{h - 1} + 1}^{\theta_h}r_{v_i}(c_{i_j})y_{i_j}}{y_{\C_h}}, \:\: \forall h \in \{1, \ldots, \eta\}.$$

Since $y_i \le \tilde{y}_i$, combining the previous inequality with Eq~(\ref{eq:opt1}), we have:
$$\sum_{h = 1}^{\eta}\mu_hy_{\C_h} \le \opt_i.$$

Since $y_{\C_h} \ge 2\sqrt{s} = \omega(1)$ for all $h \in \{1,2,\ldots, \eta\}$, we have:
$$|T_1 \cap \C_h| \le y_{\C_h} + 9 \sqrt{y_{\C_h} \log y_{\C_h}} \le \frac{3}{2} y_{\C_h}.$$ 

Since $\mu_h > r_{v_i}(c), \forall c \in \C_{h - 1}$ and since $y_{\C_h} \le 2 y_{\C_{h+1}}$, we can bound the expected score of $T_1^*$ as:
$$\sum_{c \in T_1^*}r_{v_i}(c) \leq \sum_{h = 1}^{\eta-1} \frac{3}{2} \cdot y_{\C_h}\mu_{h + 1} \le \sum_{h = 1}^{\eta-1} 3\cdot y_{\C_{h + 1}}\mu_{h + 1} \le 3\opt_i.$$

We can again assume that the $\frac{k}{\sqrt{s}}$ random candidates are chosen randomly from the entire set of $m$ candidates. This yields a bound on the score of $T_2^*$ as:
$$\sum_{c \in T_2^*}r_{v_i}(c) \leq \sum_{t = 1}^u t \cdot(\sqrt{s}\cdot \rand) = O(s^{3/2}\log s)\cdot \rand.$$
where we used $u = O(\sqrt{s\log s})$ to derive $\sum_{t = 1}^u t = O(s\log s)$. 

Therefore, we can bound $\mathsc{Good}$ as follows:
\[\mathsc{Good} \leq \sum_{c \in T_1^*}r_{v_i}(c) + \sum_{c \in T_2^*}r_{v_i}(c) \leq 3\opt_i + O(s^{3/2}\log s)\cdot \rand.\qedhere\]
\end{proof}

Synthesizing the bounds from Lemmas~\ref{lem:bad} and~\ref{lem:good}, we can conclude that:
\begin{align*}
\E[r_{v_i}(T)] &= \mathsc{Bad}\cdot \Pr[\F] + \mathsc{Good}\cdot (1-\Pr[\F])\\
&\leq  s^{-3/2} \cdot O(s^{5/2})\cdot \rand + 3\opt_i + O(s^{3/2}\log s)\cdot \rand \\
&= 3\opt_i + O(s^{3/2}\log s)\cdot \rand.
\qedhere
\end{align*}

Taking expectation over all voters, this yields Theorem~\ref{thm:lp}.

\section{Conclusion}
Our work opens some interesting directions for further research. One open question is to extend our results to the stronger notion of approximate core stability under the CC rule for which the best known result is a $16$-approximation~\cite{ChengJMW20,JiangMW20}. It would be interesting to explore if our techniques can help improve the approximation factor via a simple-to-implement procedure. 

We conjecture that there is a lower bound of $\Omega(s^{3/2}) \cdot \rand$ on the score achievable by poly-time algorithms for $s$-Borda, i.e., that the algorithm in Section~\ref{sec:lp} is almost optimal. This will require a non-trivial strengthening of known hardness results for maximum multicover~\cite{Barman}.  It would also be interesting to explore if there are greedy rules that can match these bounds.

In the same vein, another interesting question is to map the landscape of approximation ratios  for generalizations such as  committee scoring rules. The work of~\cite{Byrka} shows strong positive results when voters assign a smooth set of weights to all candidates in the committee, while our work considers the case where the weights are concentrated on higher-ranked candidates. There is a large middle ground where the approximability of this problem is poorly understood.

\section*{Acknowledgments}
We thank Brandon Fain for several discussions, and the anonymous reviewers for their suggestions. This work is supported by NSF grant CCF-1637397, ONR award N00014-19-1-2268, and DARPA award FA8650-18-C-7880.

\bibliographystyle{plain}
\bibliography{ref}

\appendix
\section{Expected Ranks in a Random Committee}
\label{app:folklore}
The following is a well-known proof of the statement that if we pick a random size-$k$ subset of $\C$, the $t^{\text{th}}$ smallest rank is $t \cdot \frac{m + 1}{k + 1}$ in expectation. Mark $m + 1$ points on a circle. Pick a subset of $k + 1$ points uniformly at random, and then choose one point $P$ of these $k + 1$ as the cut-off point uniformly at random. Starting from $P$ and going clockwise, mark the next point as the candidate with rank $1$, and the point after that as the candidate with rank $2$, and so on, until the last point which is marked as the candidate with rank $m$. The picked subset comprises $P$ and a uniformly random size-$k$ subset of $\C$. By symmetry, the expected clockwise distance going from the $t^{\text{th}}$-smallest ranked chosen candidate to the $(t + 1)^{\text{st}}$ is the same for every $t \in \{0, 1, \ldots, k\}$, if we view $P$ as simultaneously the $0^{\text{th}}$ and the $(k + 1)^{\text{st}}$ smallest. Since these $k + 1$ distances sum to $m + 1$, all of them should be $\frac{m + 1}{k + 1}$.

\section{Analysis of \g{} for $s$-Borda: Proof of Theorem~\ref{thm:s_borda}}
\label{app:sborda}
For simplicity, we define:
\[
X_k = r_{\V}(T_{k}),\qquad Y_k = \sum_{c \in \C \setminus T_k}\left(r_{\V}(T_k) - r_{\V}(T_k \cup \{c\})\right).
\]

Additionally, let $\rho_{v, k, j} $ denote the score $r_v(c_j(v))$ of the $j^{\text{th}}$-ranked candidate $c_j(v)$ for voter $v$ in the set $T_k$, and let $R_{k} = \frac{1}{n}\sum_{v \in \V} \rho_{v, k, s}$ denote the average score of these candidates across all voters. Note that by definition:
\begin{equation}
\label{eq:xk}
    X_k = r_{\V}(T_k) = \frac{1}{n} \sum_{v \in \V} \sum_{j=1}^s \rho_{v,k,j}
\end{equation}

We first present the following lemma, which is an analog to Lemma~\ref{lem:minimum_improvement}.

\begin{lemma}
\label{lem:sborda1}
For $k \geq s$, we have $X_k - X_{k + 1} \geq \frac{R_{k}^2 - (2s + 1)R_{k} + 2X_k}{2(m - k)}.$
\end{lemma}
\begin{proof}
We observe that:
\begin{align*}
Y_k + (sR_{k} - X_k) &= \sum_{c \in \C \setminus T_k} \left(r_{\V}(T_k) - r_{\V}(T_k \cup \{c\})\right) + \frac{1}{n}\sum_{v \in \V}\sum_{j = 1}^s (\rho_{v, k, s} - \rho_{v, k, j})\\
&= \frac{1}{n}\sum_{v \in \V}\sum_{i = 1}^{\rho_{v, k, s} - 1} i\\
&= \frac{\sum_{v \in V} \rho_{v, k, s}(\rho_{v, k, s} - 1)}{2n}.
\end{align*}
Here, the first equality follows from Eq~(\ref{eq:xk}). For the second equality, observe that any candidate $c \in \C \setminus T_k$ whose $r_v(c) < \rho_{v,k,s}$ contributes $\frac{1}{n} \left(\rho_{v,k,s} - r_v(c) \right)$ to the quantity $\left(r_{\V}(T_k) - r_{\V}(T_k \cup \{c\})\right)$. Therefore, the RHS of the first equality is summing, for each voter $v$, the quantity  $\frac{1}{n} \left(\rho_{v,k,s} - r_v(c) \right)$ over all $c \in \C$ whose $r_v(c) < \rho_{v,k,s}$. This yields the second equality by a change of variables.

By Cauchy-Schwarz inequality, we have:
\[
\frac{\sum_{v \in \V} \rho_{v, k, s}^2}{n} \geq \left(\frac{1}{n}\sum_{v\in \V} \rho_{v, k, s} \right)^2 = R_{k}^2.
\]
Therefore,
\[
Y_k - X_k + sR_{k} = \frac{\sum_{v \in V} \rho_{v, k, s}(\rho_{v, k, s} - 1)}{2n} \geq \frac{1}{2}R_{k}^2 - \frac{1}{2}R_{k},
\]
which is equivalent to $Y_k \geq \frac{1}{2}R_{k}^2 - \frac{2s + 1}{2}R_{k} + X_k.$

Since the candidate chosen by \g{} is at least as good as the average, we have:
\[
X_{k} - X_{k + 1} \geq \frac{Y_k}{m - k}\geq \frac{R_{k}^2 - (2s + 1)R_{k} + 2X_k}{2(m - k)}. \qedhere
\]
\end{proof}

We now present a simple relationship between $X_k$ and $R_{k}$.

\begin{lemma}
\label{lem:sborda2}
For $k \geq s$, $R_{k} \geq \frac{X_k}{s}.$
\end{lemma}
\begin{proof}
This inequality follows directly from the definition of $\rho_{v, k, s}$: since this is defined as the rank of the $s^{\text{th}}$-ranked candidate among the already-chosen ones for voter $v$, its contribution to the score must be greater than or equal to the average of the top $s$ among the chosen candidates for voter $v$. Taking sum over all voters gives this inequality.
\end{proof}

\paragraph{Completing the proof of Theorem~\ref{thm:s_borda}}
To complete the proof, we apply induction on $k$ to prove this theorem. We note that \g{} gives the optimal solution after $s$ iterations and therefore, the induction starts with $k = s$. Suppose the claim holds true for some $k - 1 \geq s$. We prove that this claim also holds true for $k$. By induction hypothesis, we have:
\[
X_{k - 1} \leq 2s^2\cdot \frac{m + 1}{k},
\]
and as in the proof for Theorem~\ref{thm:ub_greedy_1}, we only need to consider the following case:
\[
2s^2 \cdot \frac{m + 1}{k + 1} \leq X_{k - 1} \leq 2s^2 \cdot \frac{m + 1}{k}.
\]
Otherwise the induction clearly holds.

By Lemma~\ref{lem:sborda1}, we have:
\begin{align*}
X_k &\leq X_{k - 1} - \frac{R_{k - 1}^2 - (2s + 1)R_{k - 1} + 2X_{k - 1}}{2(m - k + 1)}\\ &= X_{k - 1} - \frac{X_{k - 1}}{m - k + 1} - \frac{1}{2(m - k + 1)}(R_{k - 1}^2 - (2s + 1)R_{k - 1}).
\end{align*}

Notice that the $R_{k - 1}^2 - (2s + 1)R_{k - 1}$ is a quadratic function in $R_{k - 1}$, which is monotonically increasing for $R_{k - 1} \geq \frac{2s + 1}{2}$. Since $R_{k - 1} \geq \frac{X_{k - 1}}{s} \geq 2s\cdot \frac{m + 1}{k + 1} > \frac{2s + 1}{2}$, we know $R_{k - 1}^2 - (2s + 1)R_{k - 1}$ is at least its value when $R_{k - 1} = \frac{X_{k - 1}}{s}$. Thus, we have:
\begin{align*}
X_k &\leq X_{k - 1} - \frac{X_{k - 1}}{m - k + 1} - \frac{1}{2(m - k + 1)}\left(\left(\frac{X_{k - 1}}{s}\right)^2 - (2s + 1)\frac{X_{k - 1}}{s}\right) \\
&= -\frac{1}{2(m - k + 1)s^2}X_{k - 1}^2 + \left(1 + \frac{1}{2(m - k + 1)s}\right)X_{k - 1}\\
&=-\frac{1}{2(m + 1)s^2}X_{k - 1}^2 + X_{k - 1} - \left(\frac{kX_{k - 1}}{2(m - k + 1)(m + 1)s^2} - \frac{1}{2(m - k + 1)s}\right)X_{k - 1}\\
&\leq -\frac{1}{2(m + 1)s^2}X_{k - 1}^2 + X_{k - 1}.
\end{align*}

Similar to the proof of Theorem~\ref{thm:ub_greedy_1}, notice that the right hand side is quadratic in $X_{k - 1}$ and thus monotonically increasing for $X_{k - 1} \leq (m + 1)s^2$. Since $X_{k - 1} \leq 2s^2\cdot \frac{m + 1}{k} \leq (m + 1)s^2$, the right hand side reaches its maximum at $2s^2\cdot \frac{m + 1}{k}$. Therefore,
\[
X_k \leq -\frac{1}{2(m + 1)s^2}\left(2s^2\cdot \frac{m + 1}{k}\right)^2 + 2s^2\cdot \frac{m + 1}{k} \leq 2s^2\cdot \frac{m + 1}{k + 1},
\]
concluding our induction.

\end{document}